\documentclass[aps,prl,reprint,superscriptaddress]{revtex4-2}
\usepackage{amsmath,amssymb,amsfonts,setspace,url,tikz}
\usepackage{bm}
\usepackage{latexsym}
\usepackage{amsthm}
\usepackage{float}
\usepackage{bigstrut,array,multirow}
\usepackage{here}
\usepackage{hyperref}
\usepackage{url}
\usepackage{color}
\usepackage{cleveref}
\usepackage{appendix}

\theoremstyle{plain}
\newtheorem{theorem}{Theorem}

\newtheorem{lemma}{Lemma}

\theoremstyle{definition}

\newcommand{\mathcalE}{\M}
\newcommand{\Hc}{H^\prime_{\!\bar{R}}}
\newcommand{\M}{M}
\newcommand{\tmax}{t}

\definecolor{darkgreen}{rgb}{0,.6,0}

\begin{document}
\newcommand{\Int}{\mathbb{Z}}
\newcommand{\X}{X}
\newcommand{\ket}[1]{|#1\rangle}
\newcommand{\braket}[2]{\langle#1|#2\rangle}
\newcommand{\norm}[1]{\lVert #1\rVert}
\newcommand{\abs}[1]{| #1|}
\newcommand{\bra}[1]{\langle#1|}
\newcommand{\myT}[1]{\mathcal{V}_{U,\m}^{\X}(#1)}
\newcommand{\E}{\mathbb{E}}
\newcommand{\Pp}{\mathcal{P}}
\newcommand{\Tt}{\mathcal{T}}
\newcommand{\ceil}[1]{\left\lceil #1 \right\rceil}
\newcommand{\floor}[1]{\left\lfloor #1 \right\rfloor}
\providecommand{\Aa}{\mathcal{A}}
\providecommand{\Qq}{\mathcal{Q}}
\providecommand{\Rr}{\mathsf{R}}
\providecommand{\Comp}{\mathbb{C}}
\providecommand{\Bb}{\mathcal{B}}
\providecommand{\Gin}{G^{\textrm{in}}}
\providecommand{\Vin}{V^{\textrm{in}}}
\providecommand{\Vout}{V^{\textrm{out}}}
\providecommand{\Ein}{E^{\textrm{in}}}
\providecommand{\Einter}{E^{\textrm{inter}}}
\providecommand{\degin}{\textrm{deg}_{\Ein}}
\providecommand{\degout}{\textrm{deg}_{\Eout}}
\providecommand{\Eout}{E^{\textrm{out}}}
\providecommand{\Enew}{E^{\textrm{new}}}
\providecommand{\Tr}{\mathrm{Tr}}
\newcommand{\poly}{\mathrm{poly}}
\newcommand{\mix}{\mathrm{mix}}
\newcommand{\diam}{\mathrm{diam}}
\newcommand{\ID}{\mathrm{ID}}
\newcommand{\mym}{\bar{m}}
\newcommand{\myn}{\bar{n}}
\newcommand{\set}[2]{\{#1,\ldots,#2\}}
\providecommand{\Nn}{\mathcal{N}}
\providecommand{\Vv}{\mathcal{V}}
\providecommand{\Ii}{\mathcal{I}}
\providecommand{\Cc}{\mathcal{C}}
\providecommand{\Dd}{\mathcal{D}}
\newcommand{\Hline}[1]{\noalign{\hrule height #1}} 

\author{Harry Buhrman}
\email{hbuhrman@gmail.com}
\affiliation{Quantinuum, Terrington House 13-15 Hills Road Cambridge CB2 1NL, United Kingdom}
\affiliation{QuSoft, Science Park 123,
1098 XG Amsterdam, The Netherlands}
\author{Sevag Gharibian}
\email{sevag.gharibian@upb.de}
\affiliation{Paderborn University, Department of Computer Science and Institute for Photonic Quantum Systems (PhoQS), Warburger Strasse 100, 33098 Paderborn,
Germany}
\author{Zeph Landau}
\email{zeph.landau@gmail.com}
\affiliation{University of California, Berkeley, Department of Computer Science, Berkeley, CA 94706}
\author{Fran{\c c}ois Le Gall
}
\email{legall@math.nagoya-u.ac.jp}
\affiliation{Nagoya University, Graduate School of Mathematics, Furocho, Chikusaku, 464-8602 Nagoya, Japan}
\author{Norbert Schuch}
\email{norbert.schuch@univie.ac.at}
\affiliation{University of Vienna, Faculty of Mathematics, Oskar-Morgenstern-Platz 1, 1090 Wien, Austria}
\affiliation{University of Vienna, Faculty of Physics, Boltzmanngasse 5, 1090 Wien, Austria}
\author{Suguru Tamaki}
\email{suguru.tamaki@gmail.com}
\affiliation{\mbox{University of Hyogo, Graduate School of Information Science, 8-2-1 Gakuennishi-machi, Nishi-ku, 651-2197 Kobe, Japan}}

\title{Beating the natural Grover bound for low-energy estimation and state preparation}

\begin{abstract}
    Estimating ground state energies of many-body Hamiltonians is a central task
    in many areas of quantum physics.  In this work, we give quantum algorithms
    which, given any $k$-body Hamiltonian $H$, compute an estimate for the
    ground state energy and prepare a  quantum state achieving said
    energy, respectively. Specifically, for any $\varepsilon>0$, our algorithms return, with
    high probability, an estimate of the ground state energy of $H$ within
    additive error $\varepsilon \mathcalE$, or a quantum state with
    the corresponding energy. Here,  $\M$ is the total strength of
    all interaction terms, which in general is extensive in the system size.
    Our approach makes no assumptions about the geometry or spatial locality of
    interaction terms of the input Hamiltonian and thus handles 
    even long-range or all-to-all interactions, such as in quantum chemistry, where
    lattice-based techniques break down. In this fully general setting, the runtime of our
    algorithms scales as $2^{cn/2}$ for $c<1$, yielding the first quantum
    algorithms for low-energy estimation breaking a standard square root Grover speedup for unstructured search.
    The core of
    our approach is remarkably simple, and relies on showing that 
    an extensive fraction of the
    interactions can be neglected with a controlled error. What this ultimately implies is that even arbitrary $k$-local Hamiltonians have structure in their low energy space, in the form of an exponential-dimensional low energy subspace.
\end{abstract}

\maketitle

Determining the properties of a quantum system at low energies is one of
the central problems in the study of complex quantum many-body systems. 
It is most challenging in settings with long-range or all-to-all
interactions, encountered for instance in the study of molecules
in quantum chemistry~\cite{szabo:quantum-chemistry-book,helgaker:quantum-chemistry-book},
but also in a variety of cold atomic and molecular as well as solid state
systems~\cite{ruderman:rkky, saffman:rydberg-review, islam:iontrap-longrage,
blatt:iontrap-qsim-review, yan:exchange-interact-polar-molecules,
otten:longrangeint-scondarray, gopalakrishnan:cavity-atoms-lr-interactions}.
At the same time, all-to-all interactions also appear naturally in quantum
complexity theory when encoding computationally hard tasks into quantum
Hamiltonians~\cite{kitaevClassicalQuantumComputation2002,Kempe+06}, and would necessarily~\cite{bansal+2009} have to be present in the sought-for quantum
PCP theorem~\cite{aharonovGuestColumnQuantum2013}.

The ability of both classical and quantum algorithms to compute low-energy
properties of quantum many-body systems is fundamentally limited by
quantum complexity theory, which poses restrictions on what generally
applicable algorithms can provably achieve.  Finding the ground state
energy of a system of $n$ spins up to precision $1/\mathrm{poly}(n)$ is
known to be QMA-hard (where QMA is the quantum version of NP)~\cite{kitaevClassicalQuantumComputation2002,Kempe+06,oliveirarobertoComplexityQuantumSpin2008,cubitt2016,piddockComplexityAntiferromagneticInteractions2017},
and thus hard even for
quantum computers. 

To make the problem more tractable, one can aim for some given
\emph{extensive accuracy} in terms of the total energy scale. If the total
strength of the interactions scales like the system size, as expected in
physical systems, this amounts to a constant (but arbitrarily small)
energy density,  and thus allows to study properties of the system
which are stable at finite temperature. For systems with spatially local
interactions, e.g., on a lattice, this problem can be solved efficiently
(for a fixed precision) by
cutting the system into constant-sized patches and neglecting the couplings
between those patches~\cite{bansal+2009}. However, this is no longer possible for sytems with
all-to-all interactions. In fact, the problem must remain at least
NP-hard, as the PCP theorem implies that constraint satisfaction problems
(which map to Hamiltonians with long-range couplings) remain NP-hard even
for determining whether a fixed fraction of terms is violated~\cite{PCP1,PCP2}; the
validity of a potential quantum PCP theorem~\cite{aharonovGuestColumnQuantum2013} 
would even imply its QMA-hardness.

In the light of these hardness results and the unstructured nature of the
problem, the best runtime we can expect for a classical algorithm for
estimating properties at constant energy density is exponential, by exact diagonalization methods such as the Lanczos algorithm \cite{Lanczos50,KW92}.
In fact, the Strong Exponential Time Hypothesis  states that classical algorithms \emph{must} run essentially in time scaling with $2^n$~\cite{impagliazzo2001}.
On a quantum computer, we can use Grover search, leading to a
quadratic speedup and thus a runtime of $O(2^{n/2})$ for qubit systems~\cite{grover1996,Poulin+09,Apeldoorn+20,Kerzner+23}. 
These methods can be improved if we are given access to a \emph{trial} or
\emph{guiding} state for the
low-energy subspace, as in recent quantum algorithms for ground state
energy estimation~\cite{Ge+19,Lin+20}; however, the computation of such
guiding states is a serious bottleneck, e.g., in current quantum chemistry
applications~\cite{lee+2023}. 

In this paper, we provide a quantum algorithm which provides a
super-quadratic speedup for computing properties of a general Hamiltonian
with all-to-all interactions and $k$-body terms at fixed energy density.
More specifically, we show that given a Hamiltonian $H=\sum h_\alpha$
on $(\mathbb C^d)^{\otimes n}$
with arbitrary $k$-body terms, but with no
restriction on the geometry of the interactions, there exists a quantum
algorithm with runtime 
$O(d^{(1-c_\varepsilon)\,n/2}/\varepsilon)$,
with $c_\varepsilon = \tfrac{\varepsilon}{k+\varepsilon}>0$,
which returns an estimate of the ground state energy
energy up to any extensive error $\varepsilon\, \sum\|h_\alpha\|$, as well as
a quantum algorithm which outputs a state within that energy range, from which
subsequently properties stable at low-energy density can be computed.  At
the same time, our results also yield classical algorithms for the same
task with quadratically larger runtime, which therefore in turn outperform
exact diagonalization. 

The key property which we prove and on
which our results are built is that even for systems with non-local
interactions, an extensive fraction of the
interactions can be neglected with a controlled error, just as for
spatially local systems.  This in particular also implies that the number
of states below a certain energy grows exponentially already at low
energies, again just as for
spatially local Hamiltonians---differently speaking, even all-to-all
interactions are limited in the extent to which they can rule out
low-energy states.  This allows us to construct a non-trivial guiding
state or, alternatively, to restrict the search space, and this way gives
rise to algorithms beating the natural bound based on Grover search.

\emph{Setting and summary of results.---}%
Let us now summarize our results 
in more detail. We consider $n$ $d$-level systems $(\mathbb C^d)^{\otimes n}$ which interact via a
Hamiltonian 
\begin{equation}
H=\sum_{\alpha=1}^m h_\alpha\ ,
\end{equation}
where each Hamiltonian term $h_\alpha$ is a $k$-body term, i.e., it acts on at
most $k$ of those systems, but it does not obey any spatial (geometric)
locality (hence, $m\le n^k$).  We will adopt the terminology ``$k$-local'' for such
interactions in the following, as is standard in quantum complexity theory. 
For simplicity, we will restrict to the case of qubits, $d=2$.
Define $\M := \sum_{\alpha=1}^m \|h_\alpha\|$.
Let $E_0(H)$ denote the ground state energy of $H$. We will be interested
in the task of determining $E_0(H)$ up to multiplicative accuracy
$\varepsilon$, i.e., to find an $\hat E$ with 
\begin{equation}
\label{eq:extensive-energy-error}
\vert \hat E-E_0(H)\vert \le \varepsilon\mathcalE\ ;
\end{equation}
 or to prepare a state with energy in that range. Note that for 
physical systems with geometrically local or sufficiently rapidly
decaying and uniform interactions, $\mathcalE$ will be bounded by a constant times
the system size $n$, and thus,
\eqref{eq:extensive-energy-error} amounts to an extensive accuracy, while
for a setting with $\|h_\alpha\|\le 1$, as natural for complexity
theoretic scenarios, $\mathcalE \le m$. In the following, the only assumption
on $\M$ which we will make is that it is upper bounded by a polynomial in
$n$, which in particular is satisfied in the aforementioned scenarios. 

Our first result is the construction of a new Hamiltonian $H'$ which acts
on a smaller number $n'$ of spins, and whose ground state energy $E_0(H')$
satisfies that $\vert E_0(H)-E_0(H')\vert \le \delta\mathcalE$, where
$\delta$ can be tuned by choosing $n'$; specifically, $H'$ is obtained
from $H$ by dropping the spins with the weakest couplings. This implies that any algorithm which
determines $E_0(H')$ to sufficient accuracy also gives an approximation to
$E_0(H)$ up to extensive accuracy. In particular, by suitably choosing
$n'$, this allows us to show the following (here $O^\ast$ denotes scaling up to a prefactor polynomial in~$n$): 
\begin{theorem}\label{thm:energyestimate} 
    Let $\varepsilon>0$. 
    There exists a quantum algorithm computing an $\hat E$
    satisfying $\abs{\hat E-E_0(H)}\le \varepsilon\mathcalE$ with
    high probability in time 
    \begin{equation}\label{eq:th}
    {O^\ast}\left(2^{\left(1-\frac{\varepsilon}{k+\varepsilon}\right)n/2}\,\varepsilon^{-1}\right)\ .
    \end{equation}
    There also exists a classical algorithm for this task with runtime
    ${O^\ast}\left(2^{\left(1-\frac{\varepsilon}{k+\varepsilon}\right)n}\, \varepsilon^{-1}\right)$.
\end{theorem}

In fact, we are able to show a stronger connection: Given a low-energy
state of $H'$, we can construct/prepare (in a systematic and efficient way)
low-energy states of $H$. This has two consequences. First, it gives us an
algorithm to prepare a low-energy state of $H$ in sub-Grover time, by
first preparing a low-energy state of $H'$ and then transforming it to a
low-energy state of $H$. Specifically, we show:
\begin{theorem}\label{thm:stateprep} 
    Let $\varepsilon>0$. 
    There exists a quantum algorithm preparing, with high probability,
    a mixed state $\rho$ with energy $\Tr[H\rho]\le E_0(H) +\varepsilon\mathcalE$  in time 
    \begin{equation}\label{eq:th_alg}
	O^\ast\left(2^{\left(1-\frac{\varepsilon}{2k+\varepsilon}\right)n/2}\, \varepsilon^{-1}\right).
    \end{equation}
\end{theorem}
As a second consequence, the connection between the low-energy space of
$H'$ and $H$ allows us to construct an entire low-energy subspace whose
dimension grows exponentially with energy. Specifically, we find that
for any $\varepsilon$, the number $\mathcal C(H,\varepsilon)$ of
eigenstates of $H$ with energy below $E_0(H)+\varepsilon\mathcalE$ scales as
\begin{equation}\label{th:bound-on-C}
\mathcal C(H,\varepsilon)\ge 2^{\left\lfloor\varepsilon n/(2k+\varepsilon)\right\rfloor}\ .
\end{equation}
While this scaling is expected for systems with spatially local
interactions, it is remarkable that this also holds for systems with
all-to-all interactions: It implies that even by adding arbitrary
all-to-all interactions to a spatially local Hamiltonian, one cannot raise
the energy of too many states by a too large amount. One other important
consequence  
of the bound \eqref{th:bound-on-C} is that a quantum algorithm
with a performance similar to that of Theorem~\ref{thm:energyestimate} can then also
be obtained by adapting known algorithms based on Grover search, by
exploiting that the time Grover search needs to succeed decreases for a
space which has more solutions. 

Finally, on interaction graphs of maximum degree $t$, the bounds of \Cref{thm:energyestimate}, \Cref{thm:stateprep}, and \Cref{th:bound-on-C} can in general be improved by replacing each occurence of $k$ with $(1-1/t)k$.


\emph{Construction of $H'$.---}%
We start by constructing the new Hamiltonian $H'$. We will assume
w.l.o.g.\ that all $h_\alpha$ act on exactly $k$ qubits~\footnote{Interactions
acting on less qubits can be ``padded'' by adding further qubits on which
they act trivially.}.
For each site $s=1,\dots,n$, define 
$\mathcal I(s):=\{\alpha\,\vert\, h_\alpha\mbox{\ acts on $s$}\}$ 
as the set of all interactions $\alpha\in\{1,\dots,m\}$ which act on $s$.
Further, define
\begin{equation}
e(s):=\sum_{\alpha \in \mathcal I(s)} \|h_\alpha\|\ ;
\end{equation}
it quantifies the amount of energy $e(s)$ in $H$ which originates in 
interactions which involve $s$. Importantly, 
\begin{equation}
\label{eq:sum-es-is-kE}
\sum_{s=1}^n e(s) = \sum_{\alpha=1}^m k\,\|h_\alpha\|
=k\mathcalE\ .
\end{equation}
Since the labeling $s=1,\dots,n$ of sites is arbitrary, we 
choose to label them such that $e(s)$ is monotonously increasing, 
$e(1)\le e(2)\le \cdots \le e(n)$.

Now pick some $r\in\{1,\dots,n\}$, and split the sites into the intervals 
$R:=\{1,\dots,r\}$ and $\bar R := \{r+1,\dots,n\}$. The idea is to construct
$H'$ by only keeping interactions $h_\alpha$ which lie entirely within
$\bar R$. 
Specifically, with $T$ the set of all interactions $\alpha$ which only involve sites in $\bar R$, we define
\begin{equation}
H':=\sum_{\alpha\in T}h_\alpha\ .
\end{equation}
$H'$ can either be regarded as a Hamiltonian acting on all $n$ sites, or
as acting only on the sites in $\bar R$, in which case we denote it by $\Hc$,
that is, $H'\equiv \openone^{\phantom\prime}_{\!R}\otimes \Hc$.

How well does $H'$ approximate $H$? Let $\ket\phi$ be an arbitrary state.
Then,
\begin{equation}
\label{eq:HHprime-difference-step1}
\left\vert\bra\phi H-H'\ket\phi\right\vert 
\le
\sum_{\alpha\notin T}
\left\vert\bra\phi h_\alpha \ket\phi\right\vert 
\le \sum_{\alpha\notin T} \|h_\alpha\|\ .
\end{equation}
To bound the rightmost term, note first that
\begin{equation}
\label{eq:HHprime-difference-step2}
\sum_{\alpha\notin T} \|h_\alpha\|
\le \sum_{s=1}^r e(s)\ ,
\end{equation}
as the right hand side contains each $\|h_\alpha\|$ for $\alpha\notin T$
\emph{at least} once. (On interaction graphs of maximum degree $t$, the right hand side of \Cref{eq:HHprime-difference-step2} can in general be improved by a multiplicative factor of $1-1/t$, which replaces $k$ with $(1-1/t)k$ in all subsequent runtimes stated in this paper; see Appendix for details.) Since $e(s)$ is monotonously increasing, 
\begin{equation*}
(n-r)\, e(r)\le \sum_{s=r+1}^n e(s) 
\le \sum_{s=1}^n e(s) \stackrel{\eqref{eq:sum-es-is-kE}}{=} 
k\mathcalE\ ,
\end{equation*}
which yields $e(r)\le k\mathcalE/(n-r)$. This allows us to bound 
\begin{equation*}
\sum_{s=1}^r e(s) \le r\,e(r)
\le \frac{r}{n-r} k\mathcalE\ ,
\end{equation*}
which together with 
\eqref{eq:HHprime-difference-step1}, \eqref{eq:HHprime-difference-step2}
yields the bound
\begin{equation*}
\left\vert\bra\phi H-H'\ket\phi\right\vert 
\le \frac{r}{n-r} k\mathcalE\ .
\end{equation*}
For a given $\delta>0$ \footnote{While any value of $\delta>0$ can be chosen, when $\delta$ is very small we obtain $r=0$, which does not give any speedup (since we are left with the original Hamiltonian).}, we
can thus obtain the bound
\begin{equation}
\label{eq:HHprime-difference-epsilon}
\big\vert\bra\phi H-H'\ket\phi\big\vert 
\le \delta\mathcalE
\end{equation}
on how much $H$ and $H'$ can differ
by choosing
\begin{equation}
\label{eq:HHprime-difference-r}
r = \left\lfloor \frac{\delta\,n}{k+\delta}\right\rfloor\ .
\end{equation}

From  \eqref{eq:HHprime-difference-epsilon}, we immediately obtain that
the ground state energies $E_0(H)$ and $E_0(H')$ are close: To this end,
let $\ket{\Phi'}$ be a ground state of $H'$.
Since $\bra{\Phi'} H \ket{\Phi'} - \bra{\Phi'} H'\ket{\Phi'} \le
\delta\mathcalE$, we have
\begin{equation}
\label{eq:H-Hprime-energybound-oneway}
\bra{\Phi'} H \ket{\Phi'} 
\le
\bra{\Phi'} H'\ket{\Phi'} 
+\delta\mathcalE
=
E_0(H')
+\delta\mathcalE\ .
\end{equation}
Then,
$E_0(H)\le  \bra{\Phi'} H \ket{\Phi'}  \le E_0(H') +\delta\,\mathcalE$,
and correspondingly the other way around, which yields
\begin{equation}
\label{eq:E-Eprime-difference}
\big\vert E_0(H)-E_0(H')\big\vert \le \delta\, \mathcalE\ .
\end{equation}

\emph{Construction of low-energy states and density of states.---}%
As we have observed, we can consider $H'$ as a Hamiltonian $\Hc$ 
defined on the $n-r$ qubits in $\bar R$. Let 
$\ket\varphi_{\bar R}$ be a ground state of $\Hc$. 
For an arbitrary
state $\ket\chi_R$ on the $r$ qubits in $R$, define
$\ket{\Phi''}:=\ket\chi_R\otimes \ket\varphi_{\bar R}$. Then, 
\begin{equation}
\bra{\Phi''} H'\ket{\Phi''}
=\bra\varphi \Hc \ket\varphi=E_0(H')\ ,
\end{equation}
i.e., $\ket{\Phi''}$ is a ground state of $H'$
(and as such satisfies  \eqref{eq:H-Hprime-energybound-oneway}), and thus
\begin{equation}
\bra{\Phi''}H\ket{\Phi''} 
\stackrel{\eqref{eq:H-Hprime-energybound-oneway}}{\le}
E_0(H')+\delta\mathcalE
\stackrel{\eqref{eq:E-Eprime-difference}}{\le}
E_0(H)+2\delta\mathcalE\ .
\end{equation}
Note that given $\ket{\varphi}_{\bar R}$, $\ket{\Phi''}$ can be efficiently
constructed, since all we have to do is to tensor it with an arbitrary state
$\ket{\chi}_R$ defined on the qubits in $R$.

We just saw that there is a $2^r$-dimensional space of states $V$---spanned by
all $\ket{\chi}_R\otimes \ket{\varphi}_{\bar R}$ with arbitrary
$\ket\chi_R$---with energy at most $2\delta \mathcalE$ above $E_0(H)$.
The min-max theorem (see, e.g.,~\cite{Horn+85}) then implies that 
the $2^r$ smallest eigenvalues of $H$ are all upper bounded by
$E_0(H)+2\delta\mathcalE$.
By replacing $\delta$ with $\varepsilon/2$ in
\eqref{eq:HHprime-difference-r}, we get the claimed bound
\begin{equation}
\label{eq:dos-bound-intext}
\mathcal C(H,\varepsilon)\ge 2^r= 2^{\lfloor\varepsilon n/(2k+\varepsilon)\rfloor}
\end{equation}
on the number of eigenstates with energy at most
$E_0(H)+\varepsilon\mathcalE$~\footnote{Note that Cubitt and
Gonz\'{a}lez-Guill\'{e}n~\cite{gonzalez-guillenHistorystateHamiltoniansAre2018}
previously gave exponential lower bounds on the low-energy spaces of local
Hamiltonians, but only for \emph{history state} Hamiltonians (such as
Kitaev's circuit-to-Hamiltonian
construction~\cite{kitaevClassicalQuantumComputation2002}). Our
result, in contrast, applies to any $k$-local Hamiltonian.}.

\emph{Estimating the ground energy.---}%
The bound (\ref{eq:dos-bound-intext}) can be used to improve the complexities of known ground energy estimation algorithms (e.g., the algorithm by Lin and Tong \cite{Lin+20}), by exploiting that the runtime of quantum amplitude amplification with a target space of dimension $\mathcal C(H,\varepsilon)$ improves with its inverse square root~\cite{Brassard2002}. However, we will give a simpler and more direct quantum algorithm, which moreover achieves (slightly) better performance.

The idea is simple: Recall we wish to estimate $E_0(H)$ within error $\varepsilon\mathcalE$. We will use that $E_0(H')=E_0(\Hc)$ and apply the algorithm from \cite{Lin+20} to $\Hc$, which acts on $n-r$ qubits, to obtain an estimate $\hat E'$ of $E_0(H')$. 
This in turn gives a good estimate of $E_0(H)$ by \eqref{eq:E-Eprime-difference}. 
The computational speedup comes from working on a smaller system: The runtime of the algorithm now depends on $n-r$ rather than $n$.

We now work out the details of the algorithm. 
In \eqref{eq:HHprime-difference-epsilon} and \eqref{eq:HHprime-difference-r}, we choose $\delta = (1-\frac{1}{n}) \varepsilon$.
We then compute an estimate $\hat E'$ of 
$E_0(\Hc)=E_0(H')$  such that 
\begin{equation}\label{eq3}
	\abs{\hat E'-E_0(H')}\le \frac{\varepsilon}{n} \,\mathcalE\ .
\end{equation}
This can be done in 
\begin{equation}\label{eq:qtime}
O^\ast\left(2^{(n-r)/2} \left(\frac{\varepsilon}{n}M\right)^{-1}\right)
\stackrel{\eqref{eq:HHprime-difference-r}}{=}
O^\ast\left(2^{\left(1-\frac{\varepsilon}{k+\varepsilon}\right)\frac{n}{2}}\varepsilon^{-1}\right)
\end{equation}
time by applying the quantum algorithm from \cite{Lin+20} using the completely mixed state (which has overlap $({\Tr}\left[\frac{I}{2^{n-r}}\ket{\phi}\bra{\phi}\right])^{1/2}=2^{-(n-r)/2}$ 
with the ground state of $\Hc$ and can easily be prepared) as initial state.  
Correspondingly, we can also compute $\hat E'$ on a classical computer using the Lanczos method \cite{Lanczos50,KW92} in time 
\begin{equation}\label{eq:ctime}
O^\ast\left(2^{(n-r)} \left(\frac{\varepsilon}{n}\right)^{-1}\right)
\stackrel{\eqref{eq:HHprime-difference-r}}{=}
O^\ast\left(2^{\left(1-\frac{\varepsilon}{k+\varepsilon}\right)n}\varepsilon^{-1}\right).
\end{equation}
The returned estimate will with high probability satisfy (\ref{eq3}), which together with 
\eqref{eq:E-Eprime-difference} (recall that $\delta=(1-\frac{1}{n}) \varepsilon$) guarantees that  
$\abs{\hat E'-E_0(H)}
\le \varepsilon\mathcalE$
holds, as desired.


Note that in order to apply the algorithm from \cite{Lin+20} we need to construct a block-encoding of $H'$. Since $H'$ is $k$-local, and thus sparse, 
a block-encoding can be implemented efficiently using the techniques from \cite{Gilyen+STOC19}.

\emph{Preparing a low-energy state.---}%
We now discuss how to prepare a low-energy state of $H$.
We will achieve this by first preparing a low-energy state of the Hamiltonian
$\Hc$ acting on $n-r$ qubits. As we have seen, $\Hc$ trivially embeds
into the Hamiltonian on the full $n$ qubits,
$H'\equiv\openone^{\phantom'}_R\otimes \Hc$. This embedding thus 
yields a low-energy state  for
$H'$, and thus by \eqref{eq:HHprime-difference-epsilon} also for $H$. 

To start, choose $\delta=(1-\frac{1}{n})\frac{\varepsilon}{2}$ in
(\ref{eq:HHprime-difference-epsilon},\ref{eq:HHprime-difference-r}).
Let 
\[
\Hc=\sum_{i=1}^{2^{n-r}}
\lambda_i\ket{\phi_i}\bra{\phi_i}
\]
be the eigenvalue decomposition of $\Hc$.
Using the Quantum Singular Value Transformation (QSVT) \cite{Gilyen+STOC19,Martyn+21}, we can implement in $O^*\left(\log(1/\chi)\,\varepsilon^{-1}\right)$ time an operator
\[
\Pi=\sum_{i=1}^{2^{n-r}}
P(\lambda_i)\ket{\phi_i}\bra{\phi_i},
\]
where $P\colon [-\mathcalE,\mathcalE] \to [0,1]$ is a polynomial such that 
\begin{align}
P(\lambda_i)&\ge 1-\chi&&\textrm{ if }  \lambda_i\le  E_0(H')\ ,\label{eq1}\\
P(\lambda_i)&\le\chi&&\textrm{ if }  \lambda_i\ge  E_0(H') + \frac{\varepsilon}{n}\,\mathcalE\label{eq2}\ .
\end{align}
The operator $\Pi$ is (to high accuracy) a projection onto the low-energy space of $\Hc$.
For simplicity, in the analysis below we will assume $\chi=0$. 
The QSVT framework from \cite{Gilyen+STOC19,Martyn+21} does not directly implement $\Pi$. Instead, it implements a block-encoding of $\Pi$,
i.e., a larger unitary matrix $U$ such that $\Pi$ appears (possibly after normalization)  
as the top left block of $U$:
\[
U=
\left(
\begin{matrix}
	\Pi&\cdot\\
	\cdot&\cdot
\end{matrix}
\right).
\]
The key property of this block-encoding is as follows: For any state $\sigma$ on $n-r$ qubits, applying $U$ on  $\sigma
\otimes \left(\ket{0}\bra{0}\right)^{\otimes \ell}$, where $\ell$ denotes the number of ancillas used by $U$, and subsequently measuring the ancillas in the computational basis, returns---with probability $\Tr[\Pi\sigma]$---the state 
\begin{equation}\label{eq-prepare}
\frac{\Pi\sigma\Pi }{\Tr[\Pi\sigma]}
\end{equation}
on the first $n-r$ qubits, i.e., it successfully applies $\Pi$ on $\sigma$ with probability $\Tr[\Pi\sigma]$. Due to \eqref{eq:E-Eprime-difference} and \eqref{eq2}, the state of \eqref{eq-prepare} has low energy: 
\begin{equation}\label{eq4}
\Tr\left[\Hc\left(\frac{\Pi\sigma\Pi }{\Tr[\Pi\sigma]}\right)\right]
\le
E_0(H)+\left(\delta+\frac{\varepsilon}{n}\right)\,\mathcalE\ .
\end{equation}

In order to prepare a low-energy state of $H$, we take $\sigma$ as the completely mixed state $\sigma_0$ on $n-r$ qubits,
i.e.\ $\sigma_0=\tfrac{1}{2^{n-r}}\sum \ket{\phi_i}\bra{\phi_i}$;
note that this state can be prepared efficiently.
As above, we apply the unitary $U$ and measure the ancilla qubits. The probability that we obtain 
the state \eqref{eq-prepare}  is
\begin{equation}
\label{eq_prob}
\Tr[\Pi\sigma_0]=\sum_{i=1}^{2^{n-r}}\frac{\Tr\left(\Pi\ket{\phi_i}\bra{\phi_i}\right)}{2^{n-r}}\ge
\frac{1}{2^{n-r}}\ , 
\end{equation}
where the inequality follows from \eqref{eq1}, since there exists at least
one $\ket{\phi_i}$ for which $\lambda_i=  E_0(H')$. Finally, we embed this state
in the full $n$ qubits by adding $r$ qubits each initialized to
$\ket{0}\bra{0}$:
\begin{equation}\label{eqn:rho}
\rho:=\left(\left(\ket{0}\bra{0}\right)^{\otimes r}\right)_{\! R}\otimes 
\left(\frac{\Pi\sigma_0\Pi }{\Tr[\Pi\sigma_0]}\right)_{\!\bar R}.
\end{equation}
This state satisfies Theorem \ref{thm:stateprep},
\begin{equation}\label{eqn:energybound}
\Tr[H\rho]
\stackrel{\eqref{eq:HHprime-difference-epsilon}}{\le}
\Tr[H'\rho]+\delta\mathcalE
\stackrel{\eqref{eq4}}{\le}
E_0(H) +\varepsilon\mathcalE\ 
\end{equation}
(recall that $\delta=(1-\frac{1}{n})\frac{\varepsilon}{2}$).

Using quantum amplitude amplification \cite{Brassard2002}, the probability \eqref{eq_prob} can be amplified to a probability arbitrarily close to 1 using $O(2^{(n-r)/2})$ calls to $U$. Since the unitary $U$ can be implemented in $O^\ast(\varepsilon^{-1})$ time, the overall time complexity is
\begin{equation}\label{eq:qtimeprep}
O^\ast\left(2^{(n-r)/2}\,\varepsilon^{-1}\right)\stackrel{\eqref{eq:HHprime-difference-r}}{=}
O^\ast\left(2^{\left(1-\frac{\varepsilon}{2k+\varepsilon}\right)\frac{n}{2}}\,\varepsilon^{-1}\right).
\end{equation}

\emph{Concluding remarks.---}%
In this paper we gave a quantum algorithm which provides a
super-quadratic speedup for computing properties of a general Hamiltonian
with all-to-all interactions and $k$-body terms.
The main contribution of our work is arguably not the magnitude of runtime improvement itself, but rather the following two points: First, and perhaps most surprisingly, that the standard square root speedup attainable by unstructured Grover search can be beaten, without assumptions on the local Hamiltonian itself. Second, that our approach is not heuristic, but yields rigorous and worst-case runtime guarantees. In both regards, our result can be interpreted as a quantum version of the (classical) breakthrough result by Hirsch \cite{Hirsch03}, which broke the natural bound based on brute-force search for hard classical optimization problems.

Hirsch's algorithm selects a random assignment and tries to improve it by repeatedly flipping a value of a variable chosen randomly from an unsatisfied clause (again chosen randomly). The classical analogue of our approach for computing approximately optimal solutions is similar, and slightly simpler than Hirsch: we repeatedly select a random assignment --- the bound on the running time then follows from our bound on the dimension of the low energy space. As further discussed in the Appendix, the complexity of this approach is similar to the complexity of Hirsch's algorithm.
As Hirsch's work has been influential to the development of rigorous classical approximation algorithms \cite{Escoffier+14,Alman+SODA20,Drucker20,KorhonenFOCS22,Esmer+SODA24}, we expect that our discoveries will initiate research on approximation algorithms for properties of general Hamiltonians. 

In addition, the Gap Exponential Time Hypothesis (GAP-ETH)~\cite{Dinur2016,manurangsi2017} states that there exist constants $c,\varepsilon$ such that approximating ground state energies for $3$-local Hamiltonians within error $\pm \epsilon M$ is classically impossible in $O(2^{cn})$ time. It is thus reasonable to posit a \emph{quantum} GAP-ETH, which is identical except the no-go is for quantum algorithms in $O(2^{cn/2})$ time. In this sense, one would not expect a major improvement over our Theorem ~\ref{thm:energyestimate}. 

An open question is whether our approach can be used to produce a guiding state for the ground state, to be input into the algorithms of~\cite{Ge+19,Lin+20}. On the one hand, if one believes the Quantum Strong Exponential Time Hypothesis (QSETH)~\cite{buhrman_et_al:LIPIcs.STACS.2021.19}, which states that Boolean satisfiability problems cannot be solved with superquadratic speedup over brute force search for large $k$, even by a quantum computer, then producing such a guiding state with superquadratic speedup is impossible. On the other hand, our runtime is best for small $k$, so it is not clear whether QSETH should pose an obstruction here. With this said, however, we show in the Appendix that for any $K\geq \epsilon M$, the low-energy state we produce in (\ref{eqn:rho}) has $1/K$ overlap with the space of energy at most $E_0(H) +K$. By choosing $K=(1+1/p(n))\varepsilon M$ for a sufficiently large polynomial $p$, one effectively recovers the energy bound of \Cref{eqn:energybound}, while maintaining that $\rho$ has inverse polynomial overlap with the low energy space. 

\begin{acknowledgments}
\emph{Acknowledgments.---}%
We thank Toby Cubitt for helpful discussions. 
Part of this work was done when visiting the Simons Institute for the Theory of Computing. 
SG is supported by DFG grants 432788384 and 450041824, BMBF project PhoQuant (13N16103), and project PhoQC from the State of Northrhine Westphalia.
ZL is supported by the U.S. Department of Energy, Office of Science, National Quantum Information Science Research Centers, Quantum Systems Accelerator, and by NSF Grant CCF 2311733.
FLG is supported by JSPS KAKENHI grant Nos.~JP20H05966, 24H00071 and MEXT Q-LEAP grant No.~JPMXS0120319794. 
NS is  supported by the Austrian Science Fund FWF (Grant DOIs \href{https://doi.org/10.55776/COE1}{10.55776/COE1}, \href{https://doi.org/10.55776/P36305}{10.55776/P36305}, and \href{https://doi.org/10.55776/F71}{10.55776/F71}),  the European Union -- NextGenerationEU, and the European Union’s Horizon 2020 research and innovation programme through Grant No.\ 863476 (ERC-CoG \mbox{SEQUAM}). 
ST is  supported by JSPS KAKENHI grant Nos.~JP20H05961, JP20H05967, JP22K11909.
\end{acknowledgments}

\begin{thebibliography}{45}%
\makeatletter
\providecommand \@ifxundefined [1]{%
 \@ifx{#1\undefined}
}%
\providecommand \@ifnum [1]{%
 \ifnum #1\expandafter \@firstoftwo
 \else \expandafter \@secondoftwo
 \fi
}%
\providecommand \@ifx [1]{%
 \ifx #1\expandafter \@firstoftwo
 \else \expandafter \@secondoftwo
 \fi
}%
\providecommand \natexlab [1]{#1}%
\providecommand \enquote  [1]{``#1''}%
\providecommand \bibnamefont  [1]{#1}%
\providecommand \bibfnamefont [1]{#1}%
\providecommand \citenamefont [1]{#1}%
\providecommand \href@noop [0]{\@secondoftwo}%
\providecommand \href [0]{\begingroup \@sanitize@url \@href}%
\providecommand \@href[1]{\@@startlink{#1}\@@href}%
\providecommand \@@href[1]{\endgroup#1\@@endlink}%
\providecommand \@sanitize@url [0]{\catcode `\\12\catcode `\$12\catcode `\&12\catcode `\#12\catcode `\^12\catcode `\_12\catcode `\%12\relax}%
\providecommand \@@startlink[1]{}%
\providecommand \@@endlink[0]{}%
\providecommand \url  [0]{\begingroup\@sanitize@url \@url }%
\providecommand \@url [1]{\endgroup\@href {#1}{\urlprefix }}%
\providecommand \urlprefix  [0]{URL }%
\providecommand \Eprint [0]{\href }%
\providecommand \doibase [0]{https://doi.org/}%
\providecommand \selectlanguage [0]{\@gobble}%
\providecommand \bibinfo  [0]{\@secondoftwo}%
\providecommand \bibfield  [0]{\@secondoftwo}%
\providecommand \translation [1]{[#1]}%
\providecommand \BibitemOpen [0]{}%
\providecommand \bibitemStop [0]{}%
\providecommand \bibitemNoStop [0]{.\EOS\space}%
\providecommand \EOS [0]{\spacefactor3000\relax}%
\providecommand \BibitemShut  [1]{\csname bibitem#1\endcsname}%
\let\auto@bib@innerbib\@empty
\bibitem [{\citenamefont {Szabo}\ and\ \citenamefont {Ostlund}(1996)}]{szabo:quantum-chemistry-book}%
  \BibitemOpen
  \bibfield  {author} {\bibinfo {author} {\bibfnamefont {A.}~\bibnamefont {Szabo}}\ and\ \bibinfo {author} {\bibfnamefont {N.}~\bibnamefont {Ostlund}},\ }\href@noop {} {\emph {\bibinfo {title} {Modern Quantum Chemistry: Introduction to Advanced Electronic Structure Theory}}},\ Dover Books on Chemistry\ (\bibinfo  {publisher} {Dover Publications},\ \bibinfo {year} {1996})\BibitemShut {NoStop}%
\bibitem [{\citenamefont {Helgaker}\ \emph {et~al.}(2014)\citenamefont {Helgaker}, \citenamefont {Jorgensen},\ and\ \citenamefont {Olsen}}]{helgaker:quantum-chemistry-book}%
  \BibitemOpen
  \bibfield  {author} {\bibinfo {author} {\bibfnamefont {T.}~\bibnamefont {Helgaker}}, \bibinfo {author} {\bibfnamefont {P.}~\bibnamefont {Jorgensen}},\ and\ \bibinfo {author} {\bibfnamefont {J.}~\bibnamefont {Olsen}},\ }\href@noop {} {\emph {\bibinfo {title} {Molecular Electronic-Structure Theory}}}\ (\bibinfo  {publisher} {Wiley},\ \bibinfo {year} {2014})\BibitemShut {NoStop}%
\bibitem [{\citenamefont {Ruderman}\ and\ \citenamefont {Kittel}(1954)}]{ruderman:rkky}%
  \BibitemOpen
  \bibfield  {author} {\bibinfo {author} {\bibfnamefont {M.~A.}\ \bibnamefont {Ruderman}}\ and\ \bibinfo {author} {\bibfnamefont {C.}~\bibnamefont {Kittel}},\ }\bibfield  {title} {\bibinfo {title} {Indirect exchange coupling of nuclear magnetic moments by conduction electrons},\ }\href@noop {} {\bibfield  {journal} {\bibinfo  {journal} {Physical Review}\ }\textbf {\bibinfo {volume} {96}},\ \bibinfo {pages} {99} (\bibinfo {year} {1954})}\BibitemShut {NoStop}%
\bibitem [{\citenamefont {Saffman}\ \emph {et~al.}(2010)\citenamefont {Saffman}, \citenamefont {Walker},\ and\ \citenamefont {M{\o}lmer}}]{saffman:rydberg-review}%
  \BibitemOpen
  \bibfield  {author} {\bibinfo {author} {\bibfnamefont {M.}~\bibnamefont {Saffman}}, \bibinfo {author} {\bibfnamefont {T.~G.}\ \bibnamefont {Walker}},\ and\ \bibinfo {author} {\bibfnamefont {K.}~\bibnamefont {M{\o}lmer}},\ }\bibfield  {title} {\bibinfo {title} {Quantum information with {Rydberg} atoms},\ }\href@noop {} {\bibfield  {journal} {\bibinfo  {journal} {Reviews of modern physics}\ }\textbf {\bibinfo {volume} {82}},\ \bibinfo {pages} {2313} (\bibinfo {year} {2010})}\BibitemShut {NoStop}%
\bibitem [{\citenamefont {Islam}\ \emph {et~al.}(2013)\citenamefont {Islam}, \citenamefont {Senko}, \citenamefont {Campbell}, \citenamefont {Korenblit}, \citenamefont {Smith}, \citenamefont {Lee}, \citenamefont {Edwards}, \citenamefont {Wang}, \citenamefont {Freericks},\ and\ \citenamefont {Monroe}}]{islam:iontrap-longrage}%
  \BibitemOpen
  \bibfield  {author} {\bibinfo {author} {\bibfnamefont {R.}~\bibnamefont {Islam}}, \bibinfo {author} {\bibfnamefont {C.}~\bibnamefont {Senko}}, \bibinfo {author} {\bibfnamefont {W.}~\bibnamefont {Campbell}}, \bibinfo {author} {\bibfnamefont {S.}~\bibnamefont {Korenblit}}, \bibinfo {author} {\bibfnamefont {J.}~\bibnamefont {Smith}}, \bibinfo {author} {\bibfnamefont {A.}~\bibnamefont {Lee}}, \bibinfo {author} {\bibfnamefont {E.}~\bibnamefont {Edwards}}, \bibinfo {author} {\bibfnamefont {C.-C.}\ \bibnamefont {Wang}}, \bibinfo {author} {\bibfnamefont {J.}~\bibnamefont {Freericks}},\ and\ \bibinfo {author} {\bibfnamefont {C.}~\bibnamefont {Monroe}},\ }\bibfield  {title} {\bibinfo {title} {Emergence and frustration of magnetic order with variable-range interactions in a trapped ion quantum simulator},\ }\href@noop {} {\bibfield  {journal} {\bibinfo  {journal} {Science}\ }\textbf {\bibinfo {volume} {340}},\ \bibinfo {pages} {583} (\bibinfo {year} {2013})}\BibitemShut {NoStop}%
\bibitem [{\citenamefont {Blatt}\ and\ \citenamefont {Roos}(2012)}]{blatt:iontrap-qsim-review}%
  \BibitemOpen
  \bibfield  {author} {\bibinfo {author} {\bibfnamefont {R.}~\bibnamefont {Blatt}}\ and\ \bibinfo {author} {\bibfnamefont {C.~F.}\ \bibnamefont {Roos}},\ }\bibfield  {title} {\bibinfo {title} {Quantum simulations with trapped ions},\ }\href@noop {} {\bibfield  {journal} {\bibinfo  {journal} {Nature Physics}\ }\textbf {\bibinfo {volume} {8}},\ \bibinfo {pages} {277} (\bibinfo {year} {2012})}\BibitemShut {NoStop}%
\bibitem [{\citenamefont {Yan}\ \emph {et~al.}(2013)\citenamefont {Yan}, \citenamefont {Moses}, \citenamefont {Gadway}, \citenamefont {Covey}, \citenamefont {Hazzard}, \citenamefont {Rey}, \citenamefont {Jin},\ and\ \citenamefont {Ye}}]{yan:exchange-interact-polar-molecules}%
  \BibitemOpen
  \bibfield  {author} {\bibinfo {author} {\bibfnamefont {B.}~\bibnamefont {Yan}}, \bibinfo {author} {\bibfnamefont {S.~A.}\ \bibnamefont {Moses}}, \bibinfo {author} {\bibfnamefont {B.}~\bibnamefont {Gadway}}, \bibinfo {author} {\bibfnamefont {J.~P.}\ \bibnamefont {Covey}}, \bibinfo {author} {\bibfnamefont {K.~R.}\ \bibnamefont {Hazzard}}, \bibinfo {author} {\bibfnamefont {A.~M.}\ \bibnamefont {Rey}}, \bibinfo {author} {\bibfnamefont {D.~S.}\ \bibnamefont {Jin}},\ and\ \bibinfo {author} {\bibfnamefont {J.}~\bibnamefont {Ye}},\ }\bibfield  {title} {\bibinfo {title} {Observation of dipolar spin-exchange interactions with lattice-confined polar molecules},\ }\href@noop {} {\bibfield  {journal} {\bibinfo  {journal} {Nature}\ }\textbf {\bibinfo {volume} {501}},\ \bibinfo {pages} {521} (\bibinfo {year} {2013})}\BibitemShut {NoStop}%
\bibitem [{\citenamefont {Otten}\ \emph {et~al.}(2016)\citenamefont {Otten}, \citenamefont {Rubbert}, \citenamefont {Ulrich},\ and\ \citenamefont {Hassler}}]{otten:longrangeint-scondarray}%
  \BibitemOpen
  \bibfield  {author} {\bibinfo {author} {\bibfnamefont {D.}~\bibnamefont {Otten}}, \bibinfo {author} {\bibfnamefont {S.}~\bibnamefont {Rubbert}}, \bibinfo {author} {\bibfnamefont {J.}~\bibnamefont {Ulrich}},\ and\ \bibinfo {author} {\bibfnamefont {F.}~\bibnamefont {Hassler}},\ }\bibfield  {title} {\bibinfo {title} {Universal power-law decay of electron-electron interactions due to nonlinear screening in a {Josephson} junction array},\ }\href@noop {} {\bibfield  {journal} {\bibinfo  {journal} {Physical Review B}\ }\textbf {\bibinfo {volume} {94}},\ \bibinfo {pages} {115403} (\bibinfo {year} {2016})}\BibitemShut {NoStop}%
\bibitem [{\citenamefont {Gopalakrishnan}\ \emph {et~al.}(2011)\citenamefont {Gopalakrishnan}, \citenamefont {Lev},\ and\ \citenamefont {Goldbart}}]{gopalakrishnan:cavity-atoms-lr-interactions}%
  \BibitemOpen
  \bibfield  {author} {\bibinfo {author} {\bibfnamefont {S.}~\bibnamefont {Gopalakrishnan}}, \bibinfo {author} {\bibfnamefont {B.~L.}\ \bibnamefont {Lev}},\ and\ \bibinfo {author} {\bibfnamefont {P.~M.}\ \bibnamefont {Goldbart}},\ }\bibfield  {title} {\bibinfo {title} {Frustration and glassiness in spin models with cavity-mediated interactions},\ }\href@noop {} {\bibfield  {journal} {\bibinfo  {journal} {Physical Review Letters}\ }\textbf {\bibinfo {volume} {107}},\ \bibinfo {pages} {277201} (\bibinfo {year} {2011})}\BibitemShut {NoStop}%
\bibitem [{\citenamefont {Kitaev}\ \emph {et~al.}(2002)\citenamefont {Kitaev}, \citenamefont {Shen},\ and\ \citenamefont {Vyalyi}}]{kitaevClassicalQuantumComputation2002}%
  \BibitemOpen
  \bibfield  {author} {\bibinfo {author} {\bibfnamefont {A.}~\bibnamefont {Kitaev}}, \bibinfo {author} {\bibfnamefont {A.}~\bibnamefont {Shen}},\ and\ \bibinfo {author} {\bibfnamefont {M.}~\bibnamefont {Vyalyi}},\ }\href@noop {} {\emph {\bibinfo {title} {Classical and {{Quantum Computation}}}}},\ \bibinfo {series} {Graduate {{Studies}} in {{Mathematics}}}, Vol.~\bibinfo {volume} {47}\ (\bibinfo  {publisher} {American Mathematical Society},\ \bibinfo {year} {2002})\BibitemShut {NoStop}%
\bibitem [{\citenamefont {Kempe}\ \emph {et~al.}(2006)\citenamefont {Kempe}, \citenamefont {Kitaev},\ and\ \citenamefont {Regev}}]{Kempe+06}%
  \BibitemOpen
  \bibfield  {author} {\bibinfo {author} {\bibfnamefont {J.}~\bibnamefont {Kempe}}, \bibinfo {author} {\bibfnamefont {A.~Y.}\ \bibnamefont {Kitaev}},\ and\ \bibinfo {author} {\bibfnamefont {O.}~\bibnamefont {Regev}},\ }\bibfield  {title} {\bibinfo {title} {The complexity of the local {Hamiltonian} problem},\ }\href@noop {} {\bibfield  {journal} {\bibinfo  {journal} {{SIAM} Journal on Computing}\ }\textbf {\bibinfo {volume} {35}},\ \bibinfo {pages} {1070} (\bibinfo {year} {2006})}\BibitemShut {NoStop}%
\bibitem [{\citenamefont {Bansal}\ \emph {et~al.}(2009)\citenamefont {Bansal}, \citenamefont {Bravyi},\ and\ \citenamefont {Terhal}}]{bansal+2009}%
  \BibitemOpen
  \bibfield  {author} {\bibinfo {author} {\bibfnamefont {N.}~\bibnamefont {Bansal}}, \bibinfo {author} {\bibfnamefont {S.}~\bibnamefont {Bravyi}},\ and\ \bibinfo {author} {\bibfnamefont {B.~M.}\ \bibnamefont {Terhal}},\ }\bibfield  {title} {\bibinfo {title} {Classical approximation schemes for the ground-state energy of quantum and classical {Ising} spin {Hamiltonians} on planar graphs},\ }\href@noop {} {\bibfield  {journal} {\bibinfo  {journal} {Quantum Information \& Computation}\ }\textbf {\bibinfo {volume} {9}},\ \bibinfo {pages} {701} (\bibinfo {year} {2009})}\BibitemShut {NoStop}%
\bibitem [{\citenamefont {Aharonov}\ \emph {et~al.}(2013)\citenamefont {Aharonov}, \citenamefont {Arad},\ and\ \citenamefont {Vidick}}]{aharonovGuestColumnQuantum2013}%
  \BibitemOpen
  \bibfield  {author} {\bibinfo {author} {\bibfnamefont {D.}~\bibnamefont {Aharonov}}, \bibinfo {author} {\bibfnamefont {I.}~\bibnamefont {Arad}},\ and\ \bibinfo {author} {\bibfnamefont {T.}~\bibnamefont {Vidick}},\ }\bibfield  {title} {\bibinfo {title} {Guest {{Column}}: {{The Quantum PCP Conjecture}}},\ }\href@noop {} {\bibfield  {journal} {\bibinfo  {journal} {SIGACT News}\ }\textbf {\bibinfo {volume} {44}},\ \bibinfo {pages} {47} (\bibinfo {year} {2013})}\BibitemShut {NoStop}%
\bibitem [{\citenamefont {Oliveira}\ and\ \citenamefont {Terhal}(2008)}]{oliveirarobertoComplexityQuantumSpin2008}%
  \BibitemOpen
  \bibfield  {author} {\bibinfo {author} {\bibfnamefont {R.}~\bibnamefont {Oliveira}}\ and\ \bibinfo {author} {\bibfnamefont {B.~M.}\ \bibnamefont {Terhal}},\ }\bibfield  {title} {\bibinfo {title} {The complexity of quantum spin systems on a two-dimensional square lattice},\ }\href@noop {} {\bibfield  {journal} {\bibinfo  {journal} {Quantum Information \& Computation}\ }\textbf {\bibinfo {volume} {8}},\ \bibinfo {pages} {0900} (\bibinfo {year} {2008})}\BibitemShut {NoStop}%
\bibitem [{\citenamefont {Cubitt}\ and\ \citenamefont {Montanaro}(2016)}]{cubitt2016}%
  \BibitemOpen
  \bibfield  {author} {\bibinfo {author} {\bibfnamefont {T.}~\bibnamefont {Cubitt}}\ and\ \bibinfo {author} {\bibfnamefont {A.}~\bibnamefont {Montanaro}},\ }\bibfield  {title} {\bibinfo {title} {Complexity classification of local {Hamiltonian} problems},\ }\href@noop {} {\bibfield  {journal} {\bibinfo  {journal} {SIAM Journal on Computing}\ }\textbf {\bibinfo {volume} {45}},\ \bibinfo {pages} {268} (\bibinfo {year} {2016})}\BibitemShut {NoStop}%
\bibitem [{\citenamefont {Piddock}\ and\ \citenamefont {Montanaro}(2017)}]{piddockComplexityAntiferromagneticInteractions2017}%
  \BibitemOpen
  \bibfield  {author} {\bibinfo {author} {\bibfnamefont {S.}~\bibnamefont {Piddock}}\ and\ \bibinfo {author} {\bibfnamefont {A.}~\bibnamefont {Montanaro}},\ }\bibfield  {title} {\bibinfo {title} {The complexity of antiferromagnetic interactions and {{2D}} lattices},\ }\href@noop {} {\bibfield  {journal} {\bibinfo  {journal} {Quantum Information \& Computation}\ }\textbf {\bibinfo {volume} {17}},\ \bibinfo {pages} {636} (\bibinfo {year} {2017})}\BibitemShut {NoStop}%
\bibitem [{\citenamefont {Arora}\ and\ \citenamefont {Safra}(1998)}]{PCP1}%
  \BibitemOpen
  \bibfield  {author} {\bibinfo {author} {\bibfnamefont {S.}~\bibnamefont {Arora}}\ and\ \bibinfo {author} {\bibfnamefont {S.}~\bibnamefont {Safra}},\ }\bibfield  {title} {\bibinfo {title} {Probabilistic checking of proofs: {A} new characterization of {NP}},\ }\href@noop {} {\bibfield  {journal} {\bibinfo  {journal} {Journal of the {ACM}}\ }\textbf {\bibinfo {volume} {45}},\ \bibinfo {pages} {70} (\bibinfo {year} {1998})}\BibitemShut {NoStop}%
\bibitem [{\citenamefont {Arora}\ \emph {et~al.}(1998)\citenamefont {Arora}, \citenamefont {Lund}, \citenamefont {Motwani}, \citenamefont {Sudan},\ and\ \citenamefont {Szegedy}}]{PCP2}%
  \BibitemOpen
  \bibfield  {author} {\bibinfo {author} {\bibfnamefont {S.}~\bibnamefont {Arora}}, \bibinfo {author} {\bibfnamefont {C.}~\bibnamefont {Lund}}, \bibinfo {author} {\bibfnamefont {R.}~\bibnamefont {Motwani}}, \bibinfo {author} {\bibfnamefont {M.}~\bibnamefont {Sudan}},\ and\ \bibinfo {author} {\bibfnamefont {M.}~\bibnamefont {Szegedy}},\ }\bibfield  {title} {\bibinfo {title} {Proof verification and the hardness of approximation problems},\ }\href@noop {} {\bibfield  {journal} {\bibinfo  {journal} {Journal of the {ACM}}\ }\textbf {\bibinfo {volume} {45}},\ \bibinfo {pages} {501} (\bibinfo {year} {1998})}\BibitemShut {NoStop}%
\bibitem [{\citenamefont {Lanczos}(1950)}]{Lanczos50}%
  \BibitemOpen
  \bibfield  {author} {\bibinfo {author} {\bibfnamefont {C.}~\bibnamefont {Lanczos}},\ }\bibfield  {title} {\bibinfo {title} {An iteration method for the solution of the eigenvalue problem of linear differential and integral operators},\ }\href@noop {} {\bibfield  {journal} {\bibinfo  {journal} {Journal of Research of the National Bureau of Standards}\ }\textbf {\bibinfo {volume} {45}},\ \bibinfo {pages} {255} (\bibinfo {year} {1950})}\BibitemShut {NoStop}%
\bibitem [{\citenamefont {Kuczy\'{n}ski}\ and\ \citenamefont {Wo\'{z}niakowski}(1992)}]{KW92}%
  \BibitemOpen
  \bibfield  {author} {\bibinfo {author} {\bibfnamefont {J.}~\bibnamefont {Kuczy\'{n}ski}}\ and\ \bibinfo {author} {\bibfnamefont {H.}~\bibnamefont {Wo\'{z}niakowski}},\ }\bibfield  {title} {\bibinfo {title} {Estimating the largest eigenvalue by the power and {Lanczos} algorithms with a random start},\ }\href@noop {} {\bibfield  {journal} {\bibinfo  {journal} {SIAM Journal on Matrix Analysis and Applications}\ }\textbf {\bibinfo {volume} {13}},\ \bibinfo {pages} {1094} (\bibinfo {year} {1992})}\BibitemShut {NoStop}%
\bibitem [{\citenamefont {Impagliazzo}\ and\ \citenamefont {Paturi}(2001)}]{impagliazzo2001}%
  \BibitemOpen
  \bibfield  {author} {\bibinfo {author} {\bibfnamefont {R.}~\bibnamefont {Impagliazzo}}\ and\ \bibinfo {author} {\bibfnamefont {R.}~\bibnamefont {Paturi}},\ }\bibfield  {title} {\bibinfo {title} {On the complexity of $k$-{{SAT}}},\ }\href@noop {} {\bibfield  {journal} {\bibinfo  {journal} {Journal of Computer and System Sciences}\ }\textbf {\bibinfo {volume} {62}},\ \bibinfo {pages} {367} (\bibinfo {year} {2001})}\BibitemShut {NoStop}%
\bibitem [{\citenamefont {Grover}(1996)}]{grover1996}%
  \BibitemOpen
  \bibfield  {author} {\bibinfo {author} {\bibfnamefont {L.~K.}\ \bibnamefont {Grover}},\ }\bibfield  {title} {\bibinfo {title} {A fast quantum mechanical algorithm for database search},\ }in\ \href@noop {} {\emph {\bibinfo {booktitle} {Proceedings of the 28th Annual {{ACM}} Symposium on {{Theory}} of {{Computing}}}}}\ (\bibinfo {year} {1996})\ pp.\ \bibinfo {pages} {212--219}\BibitemShut {NoStop}%
\bibitem [{\citenamefont {Poulin}\ and\ \citenamefont {Wocjan}(2009)}]{Poulin+09}%
  \BibitemOpen
  \bibfield  {author} {\bibinfo {author} {\bibfnamefont {D.}~\bibnamefont {Poulin}}\ and\ \bibinfo {author} {\bibfnamefont {P.}~\bibnamefont {Wocjan}},\ }\bibfield  {title} {\bibinfo {title} {Preparing ground states of quantum many-body systems on a quantum computer},\ }\href@noop {} {\bibfield  {journal} {\bibinfo  {journal} {Physical Review Letters}\ }\textbf {\bibinfo {volume} {102}},\ \bibinfo {pages} {130503} (\bibinfo {year} {2009})}\BibitemShut {NoStop}%
\bibitem [{\citenamefont {van Apeldoorn}\ \emph {et~al.}(2020)\citenamefont {van Apeldoorn}, \citenamefont {Gily{\'{e}}n}, \citenamefont {Gribling},\ and\ \citenamefont {de~Wolf}}]{Apeldoorn+20}%
  \BibitemOpen
  \bibfield  {author} {\bibinfo {author} {\bibfnamefont {J.}~\bibnamefont {van Apeldoorn}}, \bibinfo {author} {\bibfnamefont {A.}~\bibnamefont {Gily{\'{e}}n}}, \bibinfo {author} {\bibfnamefont {S.}~\bibnamefont {Gribling}},\ and\ \bibinfo {author} {\bibfnamefont {R.}~\bibnamefont {de~Wolf}},\ }\bibfield  {title} {\bibinfo {title} {Quantum {SDP}-{S}olvers: {B}etter upper and lower bounds},\ }\href@noop {} {\bibfield  {journal} {\bibinfo  {journal} {{Quantum}}\ }\textbf {\bibinfo {volume} {4}},\ \bibinfo {pages} {230} (\bibinfo {year} {2020})}\BibitemShut {NoStop}%
\bibitem [{\citenamefont {Kerzner}\ \emph {et~al.}(2023)\citenamefont {Kerzner}, \citenamefont {Gheorghiu}, \citenamefont {Mosca}, \citenamefont {Guilbaud}, \citenamefont {Carminati}, \citenamefont {Fracas},\ and\ \citenamefont {Dellantonio}}]{Kerzner+23}%
  \BibitemOpen
  \bibfield  {author} {\bibinfo {author} {\bibfnamefont {A.}~\bibnamefont {Kerzner}}, \bibinfo {author} {\bibfnamefont {V.}~\bibnamefont {Gheorghiu}}, \bibinfo {author} {\bibfnamefont {M.}~\bibnamefont {Mosca}}, \bibinfo {author} {\bibfnamefont {T.}~\bibnamefont {Guilbaud}}, \bibinfo {author} {\bibfnamefont {F.}~\bibnamefont {Carminati}}, \bibinfo {author} {\bibfnamefont {F.}~\bibnamefont {Fracas}},\ and\ \bibinfo {author} {\bibfnamefont {L.}~\bibnamefont {Dellantonio}},\ }\href@noop {} {\bibinfo {title} {A square-root speedup for finding the smallest eigenvalue}} (\bibinfo {year} {2023}),\ \bibinfo {note} {arXiv:2311.04379}\BibitemShut {NoStop}%
\bibitem [{\citenamefont {Ge}\ \emph {et~al.}(2019)\citenamefont {Ge}, \citenamefont {Tura},\ and\ \citenamefont {Cirac}}]{Ge+19}%
  \BibitemOpen
  \bibfield  {author} {\bibinfo {author} {\bibfnamefont {Y.}~\bibnamefont {Ge}}, \bibinfo {author} {\bibfnamefont {J.}~\bibnamefont {Tura}},\ and\ \bibinfo {author} {\bibfnamefont {J.~I.}\ \bibnamefont {Cirac}},\ }\bibfield  {title} {\bibinfo {title} {{Faster ground state preparation and high-precision ground energy estimation with fewer qubits}},\ }\href@noop {} {\bibfield  {journal} {\bibinfo  {journal} {Journal of Mathematical Physics}\ }\textbf {\bibinfo {volume} {60}},\ \bibinfo {pages} {022202} (\bibinfo {year} {2019})}\BibitemShut {NoStop}%
\bibitem [{\citenamefont {Lin}\ and\ \citenamefont {Tong}(2020)}]{Lin+20}%
  \BibitemOpen
  \bibfield  {author} {\bibinfo {author} {\bibfnamefont {L.}~\bibnamefont {Lin}}\ and\ \bibinfo {author} {\bibfnamefont {Y.}~\bibnamefont {Tong}},\ }\bibfield  {title} {\bibinfo {title} {Near-optimal ground state preparation},\ }\href@noop {} {\bibfield  {journal} {\bibinfo  {journal} {{Quantum}}\ }\textbf {\bibinfo {volume} {4}},\ \bibinfo {pages} {372} (\bibinfo {year} {2020})}\BibitemShut {NoStop}%
\bibitem [{\citenamefont {Lee}\ \emph {et~al.}(2023)\citenamefont {Lee}, \citenamefont {Lee}, \citenamefont {Zhai}, \citenamefont {Tong}, \citenamefont {Dalzell}, \citenamefont {Kumar}, \citenamefont {Helms}, \citenamefont {Gray}, \citenamefont {Cui}, \citenamefont {Liu}, \citenamefont {Kastoryano}, \citenamefont {Babbush}, \citenamefont {Preskill}, \citenamefont {Reichman}, \citenamefont {Campbell}, \citenamefont {Valeev}, \citenamefont {Lin},\ and\ \citenamefont {Chan}}]{lee+2023}%
  \BibitemOpen
  \bibfield  {author} {\bibinfo {author} {\bibfnamefont {S.}~\bibnamefont {Lee}}, \bibinfo {author} {\bibfnamefont {J.}~\bibnamefont {Lee}}, \bibinfo {author} {\bibfnamefont {H.}~\bibnamefont {Zhai}}, \bibinfo {author} {\bibfnamefont {Y.}~\bibnamefont {Tong}}, \bibinfo {author} {\bibfnamefont {A.~M.}\ \bibnamefont {Dalzell}}, \bibinfo {author} {\bibfnamefont {A.}~\bibnamefont {Kumar}}, \bibinfo {author} {\bibfnamefont {P.}~\bibnamefont {Helms}}, \bibinfo {author} {\bibfnamefont {J.}~\bibnamefont {Gray}}, \bibinfo {author} {\bibfnamefont {Z.-H.}\ \bibnamefont {Cui}}, \bibinfo {author} {\bibfnamefont {W.}~\bibnamefont {Liu}}, \bibinfo {author} {\bibfnamefont {M.}~\bibnamefont {Kastoryano}}, \bibinfo {author} {\bibfnamefont {R.}~\bibnamefont {Babbush}}, \bibinfo {author} {\bibfnamefont {J.}~\bibnamefont {Preskill}}, \bibinfo {author} {\bibfnamefont {D.~R.}\ \bibnamefont {Reichman}}, \bibinfo {author} {\bibfnamefont {E.~T.}\ \bibnamefont {Campbell}}, \bibinfo {author} {\bibfnamefont {E.~F.}\ \bibnamefont
  {Valeev}}, \bibinfo {author} {\bibfnamefont {L.}~\bibnamefont {Lin}},\ and\ \bibinfo {author} {\bibfnamefont {G.~K.-L.}\ \bibnamefont {Chan}},\ }\bibfield  {title} {\bibinfo {title} {Evaluating the evidence for exponential quantum advantage in ground-state quantum chemistry},\ }\href@noop {} {\bibfield  {journal} {\bibinfo  {journal} {Nature Communications}\ }\textbf {\bibinfo {volume} {14}},\ \bibinfo {pages} {1952} (\bibinfo {year} {2023})}\BibitemShut {NoStop}%
\bibitem [{Note1()}]{Note1}%
  \BibitemOpen
  \bibinfo {note} {Interactions acting on less qubits can be ``padded'' by adding further qubits on which they act trivially.}\BibitemShut {Stop}%
\bibitem [{Note2()}]{Note2}%
  \BibitemOpen
  \bibinfo {note} {While any value of $\delta >0$ can be chosen, when $\delta $ is very small we obtain $r=0$, which does not give any speedup (since we are left with the original Hamiltonian).}\BibitemShut {Stop}%
\bibitem [{\citenamefont {Horn}\ and\ \citenamefont {Johnson}(1985)}]{Horn+85}%
  \BibitemOpen
  \bibfield  {author} {\bibinfo {author} {\bibfnamefont {R.~A.}\ \bibnamefont {Horn}}\ and\ \bibinfo {author} {\bibfnamefont {C.~R.}\ \bibnamefont {Johnson}},\ }\href@noop {} {\emph {\bibinfo {title} {Matrix Analysis}}}\ (\bibinfo  {publisher} {Cambridge University Press},\ \bibinfo {year} {1985})\BibitemShut {NoStop}%
\bibitem [{Note3()}]{Note3}%
  \BibitemOpen
  \bibinfo {note} {Note that Cubitt and Gonz\'{a}lez-Guill\'{e}n~\cite {gonzalez-guillenHistorystateHamiltoniansAre2018} previously gave exponential lower bounds on the low-energy spaces of local Hamiltonians, but only for \protect \emph {history state} Hamiltonians (such as Kitaev's circuit-to-Hamiltonian construction~\cite {kitaevClassicalQuantumComputation2002}). Our result, in contrast, applies to any $k$-local Hamiltonian.}\BibitemShut {Stop}%
\bibitem [{\citenamefont {Brassard}\ \emph {et~al.}(2002)\citenamefont {Brassard}, \citenamefont {H{\o}yer}, \citenamefont {Mosca},\ and\ \citenamefont {Tapp}}]{Brassard2002}%
  \BibitemOpen
  \bibfield  {author} {\bibinfo {author} {\bibfnamefont {G.}~\bibnamefont {Brassard}}, \bibinfo {author} {\bibfnamefont {P.}~\bibnamefont {H{\o}yer}}, \bibinfo {author} {\bibfnamefont {M.}~\bibnamefont {Mosca}},\ and\ \bibinfo {author} {\bibfnamefont {A.}~\bibnamefont {Tapp}},\ }\bibfield  {title} {\bibinfo {title} {Quantum amplitude amplification and estimation},\ }\href@noop {} {\bibfield  {journal} {\bibinfo  {journal} {Contemporary Mathematics}\ }\textbf {\bibinfo {volume} {305}},\ \bibinfo {pages} {53} (\bibinfo {year} {2002})}\BibitemShut {NoStop}%
\bibitem [{\citenamefont {Gily{\'{e}}n}\ \emph {et~al.}(2019)\citenamefont {Gily{\'{e}}n}, \citenamefont {Su}, \citenamefont {Low},\ and\ \citenamefont {Wiebe}}]{Gilyen+STOC19}%
  \BibitemOpen
  \bibfield  {author} {\bibinfo {author} {\bibfnamefont {A.}~\bibnamefont {Gily{\'{e}}n}}, \bibinfo {author} {\bibfnamefont {Y.}~\bibnamefont {Su}}, \bibinfo {author} {\bibfnamefont {G.~H.}\ \bibnamefont {Low}},\ and\ \bibinfo {author} {\bibfnamefont {N.}~\bibnamefont {Wiebe}},\ }\bibfield  {title} {\bibinfo {title} {Quantum singular value transformation and beyond: exponential improvements for quantum matrix arithmetics},\ }in\ \href@noop {} {\emph {\bibinfo {booktitle} {Proceedings of the 51st Annual {ACM} {SIGACT} Symposium on Theory of Computing}}}\ (\bibinfo {year} {2019})\ pp.\ \bibinfo {pages} {193--204}\BibitemShut {NoStop}%
\bibitem [{\citenamefont {Martyn}\ \emph {et~al.}(2021)\citenamefont {Martyn}, \citenamefont {Rossi}, \citenamefont {Tan},\ and\ \citenamefont {Chuang}}]{Martyn+21}%
  \BibitemOpen
  \bibfield  {author} {\bibinfo {author} {\bibfnamefont {J.~M.}\ \bibnamefont {Martyn}}, \bibinfo {author} {\bibfnamefont {Z.~M.}\ \bibnamefont {Rossi}}, \bibinfo {author} {\bibfnamefont {A.~K.}\ \bibnamefont {Tan}},\ and\ \bibinfo {author} {\bibfnamefont {I.~L.}\ \bibnamefont {Chuang}},\ }\bibfield  {title} {\bibinfo {title} {Grand unification of quantum algorithms},\ }\href@noop {} {\bibfield  {journal} {\bibinfo  {journal} {PRX Quantum}\ }\textbf {\bibinfo {volume} {2}},\ \bibinfo {pages} {040203} (\bibinfo {year} {2021})}\BibitemShut {NoStop}%
\bibitem [{\citenamefont {Hirsch}(2003)}]{Hirsch03}%
  \BibitemOpen
  \bibfield  {author} {\bibinfo {author} {\bibfnamefont {E.~A.}\ \bibnamefont {Hirsch}},\ }\bibfield  {title} {\bibinfo {title} {Worst-case study of local search for {MAX-$k$-SAT}},\ }\href@noop {} {\bibfield  {journal} {\bibinfo  {journal} {Discrete Applied Mathematics}\ }\textbf {\bibinfo {volume} {130}},\ \bibinfo {pages} {173} (\bibinfo {year} {2003})}\BibitemShut {NoStop}%
\bibitem [{\citenamefont {Escoffier}\ \emph {et~al.}(2014)\citenamefont {Escoffier}, \citenamefont {Paschos},\ and\ \citenamefont {Tourniaire}}]{Escoffier+14}%
  \BibitemOpen
  \bibfield  {author} {\bibinfo {author} {\bibfnamefont {B.}~\bibnamefont {Escoffier}}, \bibinfo {author} {\bibfnamefont {V.~T.}\ \bibnamefont {Paschos}},\ and\ \bibinfo {author} {\bibfnamefont {E.}~\bibnamefont {Tourniaire}},\ }\bibfield  {title} {\bibinfo {title} {Approximating {MAX} {SAT} by moderately exponential and parameterized algorithms},\ }\href@noop {} {\bibfield  {journal} {\bibinfo  {journal} {Theoretical Computer Science}\ }\textbf {\bibinfo {volume} {560}},\ \bibinfo {pages} {147} (\bibinfo {year} {2014})}\BibitemShut {NoStop}%
\bibitem [{\citenamefont {Alman}\ \emph {et~al.}(2020)\citenamefont {Alman}, \citenamefont {Chan},\ and\ \citenamefont {Williams}}]{Alman+SODA20}%
  \BibitemOpen
  \bibfield  {author} {\bibinfo {author} {\bibfnamefont {J.}~\bibnamefont {Alman}}, \bibinfo {author} {\bibfnamefont {T.~M.}\ \bibnamefont {Chan}},\ and\ \bibinfo {author} {\bibfnamefont {R.~R.}\ \bibnamefont {Williams}},\ }\bibfield  {title} {\bibinfo {title} {Faster deterministic and las vegas algorithms for offline approximate nearest neighbors in high dimensions},\ }in\ \href@noop {} {\emph {\bibinfo {booktitle} {Proceedings of the 2020 {ACM-SIAM} Symposium on Discrete Algorithms}}}\ (\bibinfo {year} {2020})\ pp.\ \bibinfo {pages} {637--649}\BibitemShut {NoStop}%
\bibitem [{\citenamefont {Drucker}(2020)}]{Drucker20}%
  \BibitemOpen
  \bibfield  {author} {\bibinfo {author} {\bibfnamefont {A.}~\bibnamefont {Drucker}},\ }\bibfield  {title} {\bibinfo {title} {An improved exponential-time approximation algorithm for fully-alternating games against nature},\ }in\ \href@noop {} {\emph {\bibinfo {booktitle} {Proceedings of the 61st {IEEE} Annual Symposium on Foundations of Computer Science}}}\ (\bibinfo {year} {2020})\ pp.\ \bibinfo {pages} {1081--1090}\BibitemShut {NoStop}%
\bibitem [{\citenamefont {Korhonen}(2021)}]{KorhonenFOCS22}%
  \BibitemOpen
  \bibfield  {author} {\bibinfo {author} {\bibfnamefont {T.}~\bibnamefont {Korhonen}},\ }\bibfield  {title} {\bibinfo {title} {A single-exponential time 2-approximation algorithm for treewidth},\ }in\ \href@noop {} {\emph {\bibinfo {booktitle} {Proceedings of the 62nd {IEEE} Annual Symposium on Foundations of Computer Science}}}\ (\bibinfo {year} {2021})\ pp.\ \bibinfo {pages} {184--192}\BibitemShut {NoStop}%
\bibitem [{\citenamefont {Esmer}\ \emph {et~al.}(2024)\citenamefont {Esmer}, \citenamefont {Kulik}, \citenamefont {Marx}, \citenamefont {Neuen},\ and\ \citenamefont {Sharma}}]{Esmer+SODA24}%
  \BibitemOpen
  \bibfield  {author} {\bibinfo {author} {\bibfnamefont {B.~C.}\ \bibnamefont {Esmer}}, \bibinfo {author} {\bibfnamefont {A.}~\bibnamefont {Kulik}}, \bibinfo {author} {\bibfnamefont {D.}~\bibnamefont {Marx}}, \bibinfo {author} {\bibfnamefont {D.}~\bibnamefont {Neuen}},\ and\ \bibinfo {author} {\bibfnamefont {R.}~\bibnamefont {Sharma}},\ }\bibfield  {title} {\bibinfo {title} {Optimally repurposing existing algorithms to obtain exponential-time approximations},\ }in\ \href@noop {} {\emph {\bibinfo {booktitle} {Proceedings of the 2024 {ACM-SIAM} Symposium on Discrete Algorithms}}}\ (\bibinfo {year} {2024})\ pp.\ \bibinfo {pages} {314--345}\BibitemShut {NoStop}%
\bibitem [{\citenamefont {Dinur}(2016)}]{Dinur2016}%
  \BibitemOpen
  \bibfield  {author} {\bibinfo {author} {\bibfnamefont {I.}~\bibnamefont {Dinur}},\ }\bibfield  {title} {\bibinfo {title} {Mildly exponential reduction from gap-{3SAT} to polynomial-gap label-cover},\ }\href@noop {} {\bibfield  {journal} {\bibinfo  {journal} {Electronic colloquium on computational complexity ECCC ; research reports, surveys and books in computational complexity}\ } (\bibinfo {year} {2016})}\BibitemShut {NoStop}%
\bibitem [{\citenamefont {Manurangsi}\ and\ \citenamefont {Raghavendra}(2017)}]{manurangsi2017}%
  \BibitemOpen
  \bibfield  {author} {\bibinfo {author} {\bibfnamefont {P.}~\bibnamefont {Manurangsi}}\ and\ \bibinfo {author} {\bibfnamefont {P.}~\bibnamefont {Raghavendra}},\ }\bibfield  {title} {\bibinfo {title} {{A Birthday Repetition Theorem and Complexity of Approximating Dense CSPs}},\ }in\ \href@noop {} {\emph {\bibinfo {booktitle} {Proceedings of the 44th International Colloquium on Automata, Languages, and Programming (ICALP 2017)}}},\ \bibinfo {series} {Leibniz International Proceedings in Informatics (LIPIcs)}, Vol.~\bibinfo {volume} {80}\ (\bibinfo {year} {2017})\ pp.\ \bibinfo {pages} {78:1--78:15}\BibitemShut {NoStop}%
\bibitem [{\citenamefont {Buhrman}\ \emph {et~al.}(2021)\citenamefont {Buhrman}, \citenamefont {Patro},\ and\ \citenamefont {Speelman}}]{buhrman_et_al:LIPIcs.STACS.2021.19}%
  \BibitemOpen
  \bibfield  {author} {\bibinfo {author} {\bibfnamefont {H.}~\bibnamefont {Buhrman}}, \bibinfo {author} {\bibfnamefont {S.}~\bibnamefont {Patro}},\ and\ \bibinfo {author} {\bibfnamefont {F.}~\bibnamefont {Speelman}},\ }\bibfield  {title} {\bibinfo {title} {{A Framework of Quantum Strong Exponential-Time Hypotheses}},\ }in\ \href@noop {} {\emph {\bibinfo {booktitle} {Proceedings of the 38th International Symposium on Theoretical Aspects of Computer Science (STACS 2021)}}},\ \bibinfo {series} {Leibniz International Proceedings in Informatics (LIPIcs)}, Vol.\ \bibinfo {volume} {187}\ (\bibinfo {year} {2021})\ pp.\ \bibinfo {pages} {19:1--19:19}\BibitemShut {NoStop}%
\bibitem [{\citenamefont {{Gonz{\'a}lez-Guill{\'e}n}}\ and\ \citenamefont {Cubitt}(2018)}]{gonzalez-guillenHistorystateHamiltoniansAre2018}%
  \BibitemOpen
  \bibfield  {author} {\bibinfo {author} {\bibfnamefont {C.~E.}\ \bibnamefont {{Gonz{\'a}lez-Guill{\'e}n}}}\ and\ \bibinfo {author} {\bibfnamefont {T.~S.}\ \bibnamefont {Cubitt}},\ }\href@noop {} {\bibinfo {title} {History-state {{Hamiltonians}} are critical}} (\bibinfo {year} {2018}),\ \bibinfo {note} {arXiv:1810.06528}\BibitemShut {NoStop}%
\end{thebibliography}
\providecommand{\noopsort}[1]{}
%

\appendix
\newcommand{\ER}{E(R)}
\newcommand{\ERb}{E'(R)}
\newcommand{\Rb}{\overline{R}}
\newcommand{\eg}[2]{E_{#1,#2}}
\newcommand{\ag}[2]{A_{#1,#2}}
\newcommand{\hg}[1]{H_{#1,#1}}
\newcommand{\rg}[2]{R_{#1,#2}}

\appendix
\renewcommand{\appendixpagename}{Appendix}
\appendixpage
In this Appendix, we first show how to improve all runtimes and key bounds in our paper by taking the maximum degree $\tmax$ of the input Hamiltonian's interaction hypergraph into account. The improvements apply to \Cref{eq:HHprime-difference-r} (bound on $r$), \Cref{eq:dos-bound-intext} (bound on dimension of low energy space), \Cref{eq:qtime} (quantum runtime for low energy estimation), \Cref{eq:ctime} (classical runtime for low energy estimation), and \Cref{eq:qtimeprep} (quantum runtime for preparation of low energy state). Our approach is to strengthen \Cref{eq:HHprime-difference-step2}, which we do in Section A. Section B then shows how an improved \Cref{eq:HHprime-difference-step2} in turn improves all bounds mentioned above. 

Second, we show in Section C that the state produced in \Cref{eqn:rho} constitutes a guiding state for the low energy space of $H$.

Finally, in Section D we discuss the relation between our approach and Hirsch's work.

\section{A. Improving Equation (\ref{eq:HHprime-difference-step2})}

\noindent\emph{Definitions.} Define a $k$-hypergraph $G=(V,E)$ as a hypergraph with all hyperedges of size exactly $k$. By assumption, $H$'s interaction hypergraph is thus a $k$-hypergraph. For any vertex $v\in V$, its \emph{degree} $\deg(v)$ is the number of hyperedges containing $v$. Two vertices $v,w\in V$ are neighbors if there exists $e\in E$ with $u,v\in e$. Note that the number of neighbors for any $v$ is at most $\deg(v)(k-1)$. For any $v,w\in V$, the \emph{distance} $d(u,v)$ between $v$ and $w$ is the minimum number $m$ of distinct hyperedges $(e_1,\ldots, e_m)$ such that $v\in e_1$, $w\in e_m$, and for all $i\in [m-1]$, $e_i$ and $e_{i+1}$ share at least one common vertex. Define $l_i(v)$ as the set of vertices at distance less than or equal to $i$ from $v$. Observing that $l_1(v)$ is just the neighbor set of $v$, we thus have
\begin{equation}\label{eqn:dist1}
    \abs{l_1(v)}\leq \deg(v)(k-1).
\end{equation}
    Finally, for any $R\subseteq V$, let $E(R)$ denote the subhypergraph of $G$ induced by $R$. In words, these are the hyperedges internal to $R$.

\begin{lemma} \label{l:strongerlemma}
    Consider any $k$-local Hamiltonian with unit weight edges and maximum degree $\tmax$. Among all vertices of degree $t$, let $v$ have minimal neighbor set size, i.e., minimizing $\abs{l_1(v)}$. Then, for any $r\geq \abs{l_1(v)}+1$, 
    \begin{equation}\label{eq:improved}
    \sum_{\alpha\notin T} \|h_\alpha\|
    \le \frac{1}{1+\beta(k-1)}\sum_{s=1}^r e(s),
    \end{equation}
    where
    \begin{align}
        \beta:=\frac{1}{(t-1)(k-1)}\cdot \frac{\tmax + \lfloor\gamma\rfloor}{\tmax+\gamma}\label{eqn:beta}
    \end{align} 
    with $\gamma:=\frac{r-\abs{l_1(v)}-1}{k-1}$. Recall by \Cref{eqn:dist1} that $\abs{l_1(v)}\leq \deg(v)(k-1)$. 
\end{lemma}
\noindent Before proving \Cref{l:strongerlemma}, let us digest it and discuss tightness. The simplest setting is $k=2$, yielding $\beta=1/(t-1)$, and simplifying \Cref{eq:improved} to
\begin{align}
    \sum_{\alpha\notin T} \|h_\alpha\|
    \le \left(1-\frac{1}{t}\right)\sum_{s=1}^r e(s),\label{eqn:ideal}
\end{align}
For $k>2$, ``morally'' the lower bound for $\beta$ should be
\begin{align}
    \beta\geq\frac{1}{(t-1)(k-1)},\label{eqn:close}
\end{align}
and this holds for any $r$ such that $r-\abs{l_1(v)}-1$ is divisible by $k-1$, which recovers \Cref{eqn:ideal}. If $r-\abs{l_1(v)}-1$ is not divisible by $k-1$, on the other hand, $\beta$ can slightly drop below the bound in \Cref{eqn:close}, but the loss is vanishing in $r$. In other words, since $k\in O(1)$, when $t\in o(r)$ (which is the case, for example, when $t\in O(1)$ and $r\in\Theta(n)$, the latter of which is generally the case for our algorithm), then as $r\rightarrow \infty$, we recover \Cref{eqn:close}. Put another way, up to adding a constant to $r$, we may assume $r-\abs{l_1(v)}-1$ is indeed divisible by $k-1$, and thus \Cref{eqn:close} holds.

Moving on to tightness of \Cref{l:strongerlemma}, first, the floor function in \Cref{eqn:beta} cannot be dropped, as it stems from our particular proof approach. It is, in particular, easy to construct examples for carefully chosen $r$ so that for our construction, the version of \Cref{eqn:beta} without the floor is violated. Second, the setting of $\beta$ in \Cref{eqn:beta} is for a worst case analysis --- in many hypergraphs, larger values of $\beta$ can be attained. For example, for $2$-local Hamiltonians on $t$-regular graphs with a cycle of some length $c$, setting $r=c$ yields $\beta=1/(t-2)$ in our analysis. For $3$-regular interaction graphs, for example, this improves \Cref{eqn:ideal} to
\begin{align}
    \sum_{\alpha\notin T} \|h_\alpha\|
    \le \frac{1}{2}\sum_{s=1}^r e(s),
\end{align}
which is tight. (More generally, an identical statement holds if for a given $r$, one can efficiently find a set of cycles in $G$ whose union has exactly $r$ vertices.) Similar statements can be derived by instead considering paths of length $r$ in $G$, for $r$ larger than a constant.

\begin{proof}[Proof of \Cref{l:strongerlemma}]
    \noindent In the previous construction, we ordered vertices according to non-decreasing $e(s)$, and set $R=\{1,\ldots, r\}$. To show the improved bound, observe that in unit weight graphs, the sum over $\alpha\not\in T$ counts the number of hyperedges either contained within $R$, or crossing from $R$ to its complement. If all hyperedges fall into the latter case, the equality $\sum_{\alpha\notin T} \|h_\alpha\|
    = \sum_{s=1}^r e(s)$ holds, whereas at the other extreme, if all hyperedges are contained in $R$, the equality $\sum_{\alpha\notin T} \|h_\alpha\|
    = \frac{1}{k}\sum_{s=1}^r e(s)$ holds. We show that by picking $R$ greedily to maximize edges interal to $R$, we can always achieve a multiplicative factor of ${1}/(1+\beta(k-1))$.

    \emph{To choose $R$:}  
    We construct $R$ as follows (intuition to follow):
    \begin{enumerate}
        \item Set $R=\{v\}$ and $i=1$.
        \item While $l_i(v)\leq r$ do //Add full levels while possible
            \begin{enumerate} 
                \item Set $R=R\cup l_i(v)$.
                \item Set $r=r-\abs{l_i(v)}$.
                \item Set $i=i+1$.
            \end{enumerate}
            \item Repeat: //Greedily choose vertices from $l_i(v)$
            \begin{enumerate}
                \item Let $S\subseteq l_i(v)\setminus R$ be of minimal size so that $\abs{E(R\cup S)}>\abs{E(R)}$ and $\abs{S}\leq r$.
                \item If no such $S$ exists then
                \begin{enumerate}
                    \item Add any $r$ vertices from $l_i(v)\setminus R$ to $R$.
                    \item Exit loop.
                \end{enumerate}
                \item Set $R=R\cup S$.
                \item Set $r=r-\abs{S}$.
            \end{enumerate} 
   \end{enumerate} 
    In words, we conduct a modified breadth-first search on $G$ starting at $v$. For as long as we can, we add all vertices on each level to $R$ (loop on line 2). Once $\abs{l_i(v)}<r$, however, we need to be more careful (loop on line $3$). We repeatedly choose the smallest subset $S$, so that each $S$ added to $R$ increases the number of hyperedges internal to $R$ by at least $1$. Note line 3(a) can be done efficiently, since the maximum size of any $S$ to be considered is at most $k-1$ for $k\in O(1)$. This is because, by definition of breadth-first search, for any $i>1$, we can visit a node $w$ at level $i$ via hyperedge $e$ if and only if we visited a neighbor $u\in e$ of $w$ in level $i-1$. This implies \emph{all} vertices of $e$ must be contained in levels $i-1$ and $i$. By line (2), all vertices of $e$ in level $i-1$ have been added to $R$. Thus, there are at most $k-1$ vertices of $e$ in level $i$, whose addition to $R$ adds $e$ to $\ER$.

\begin{lemma}\label{l:ER}
    Let $\ERb$ denote the set of edges with at least one vertex in each of $R$ and $\overline{R}$. Then, for any $r\geq 1 + \abs{l_1(v)}$,
    \begin{equation}\label{eqn:erbound}
        \frac{\abs{\ER}}{\abs{\ERb}}\geq \frac{\tmax + \lfloor\frac{r-\abs{l_1(v)}-1}{k-1}\rfloor}{(\tmax-1)(\tmax(k-1)+r-\abs{l_1(v)}-1)}.
    \end{equation}
\end{lemma}
\begin{proof}[Proof of \Cref{l:ER}]
    We proceed in two steps: (1) The first $1+\abs{l_1(v)}$ vertices added to $R$ by construction satisfy a certain ``base'' ratio. (2) Each additional vertex added to $R$ thereafter roughly preserves the base ratio. 
    
    For (1), the first $1+\abs{l_1(v)}$ vertices added to $R$ are $v$ and all its neighbors. At this point, $\abs{E(R)}=t$ since $v$ has degree $t$. And $\abs{\ERb}\leq t(k-1)(t-1)$, since by \Cref{eqn:dist1}, $v$ has at most $t(k-1)$ neighbors, each of which has at most $t-1$ neighbors distinct from $v$. Thus, at this point we have base ratio
    \begin{align}
        \frac{\abs{\ER}}{\abs{\ERb}}\geq\frac{t}{t(k-1)(t-1)}.
    \end{align}
    For (2), recall from above that any vertex $w$ at level $i$ visited via hyperedge $e$ has all its neighbors in levels $i-1$ and $i$. By line 2, all neighbors in level $i-1$ are already in $R$. From this, if the condition of line 2 holds, then line 2(a) adds all vertices in level $i$ to $R$, and this includes the at most $k-1$ remaining neighbors of $w$ (including $w$), thus adding $w$ to $\ER$. The case of line 3 is similar, but additionally uses the fact that the greedy heuristic of line 3(a) guarantees that, each time we pick a subset $S$ of at most $k-1$ vertices at level $i$, at least one new hyperedge $e$ is added to $\ER$. We conclude in either case that, after the first $1+\abs{l_1(v)}$ vertices are added to $R$, it takes at most $k-1$ new vertices being added to $R$ in order to add a new hyperedge to $\ER$. This yields the additive term $\lfloor\frac{r-\abs{l_1(v)}-1}{k-1}\rfloor$ in the numerator of \Cref{eqn:erbound}. 
    
    As for the denominator, we proceed by induction. For the base case, all hyperedges containing $v$ are in $\ER$ by construction. Thus, for any vertex $w$ in $l_1(v)$ added to $R$, we add at most $t-1$ new hyperedges to $\ERb$, since at least one hyperedge contains both $v$ and $w$, which is in $\ER$. For the inductive step, consider any vertex $w$ added after the first $1+\abs{l_1(v)}$ vertices were added to $R$. Again, the addition of $w$ to $R$ adds at most $t-1$ new hyperedges to $\ERb$, but for a slightly different reason. Recall that in order to visit $w$ on level $i$ via hyperedge $e$, the breadth-first search must first have visited a neighbor $u\in e$ of $w$ on level $i-1$. But by induction, $e$ was already added to $\ERb$ when $u$ was added to $R$. This yields the additive term $(t-1)(r-\abs{l_1(v)-1})$ in the denominator of \Cref{eqn:erbound}, completing the proof.
\end{proof}
With \Cref{l:ER} in hand, we can prove \Cref{eq:improved}. We have
    \begin{align}
        \sum_{\alpha\notin T} \|h_\alpha\| &= \abs{\ER}+\abs{\ERb}=:m.\label{eq:first}\\
        \sum_{s=1}^r e(s) &= k\abs{\ER}+\abs{\ERb}.
    \end{align}
    It suffices to consider only the boundary case of \Cref{l:ER} with 
    \begin{align}
        \frac{\abs{\ER}}{\abs{\ERb}}= \frac{\tmax + \lfloor\frac{r-\abs{l_1(v)}-1}{k-1}\rfloor}{(\tmax-1)(\tmax(k-1)+r-\abs{l_1(v)}-1)}=\beta,
    \end{align} 
    as $\sum_{s=1}^r e(s)$ increases monotonically with increasing $\ER/\ERb$. Then:
    \begin{align}
        \sum_{s=1}^r e(s) &= 
        \beta k m+(1-\beta)m\\
        \label{eq:last}   &= (1+\beta(k-1))m.
    \end{align}
    Combining \Cref{eq:first} with \Cref{eq:last} yields \Cref{eq:improved}. 
\end{proof}

\section{B. Applications of an improved Equation (\ref{eq:HHprime-difference-step2})}\label{app:application}

Defining $\Delta:=\frac{1}{1+\beta(k-1)}$ in \Cref{eq:improved}, we first obtain an improved bound for \Cref{eq:HHprime-difference-r}:
\begin{equation}
    r = \left\lfloor \frac{\delta\,n}{\Delta k+\delta}\right\rfloor\ .
\end{equation}
In turn, this improves \Cref{eq:dos-bound-intext} to
\begin{equation}
    \mathcal C(H,\varepsilon)\ge 2^{\lfloor\varepsilon n/(2\Delta k+\varepsilon)\rfloor}, 
    \end{equation}
    \Cref{eq:qtime} to
\begin{equation}
    O^\ast\left(2^{\left(1-\frac{\varepsilon}{\Delta k+\varepsilon}\right)\frac{n}{2}}\varepsilon^{-1}\right),
\end{equation}
and analogously for \Cref{eq:ctime}. Finally, \Cref{eq:qtimeprep} is improved to 
\begin{equation}
    O^\ast\left(2^{\left(1-\frac{\varepsilon}{2\Delta k+\varepsilon}\right)\frac{n}{2}}\,\varepsilon^{-1}\right).
\end{equation}
Thus, for example, when $k=2$, or when $k>3$ with any $r$ such that $r-\abs{l_1(v)}-1$ is divisible by $k-1$, we have for $t=3$-regular hypergraphs that
\begin{equation}
    r = \left\lfloor \frac{\delta\,n}{\frac{2}{3}k+\delta}\right\rfloor\ ,
\end{equation}
and thus \Cref{eq:dos-bound-intext} becomes
\begin{equation}
    \mathcal C(H,\varepsilon)\ge 2^{\lfloor\varepsilon n/(\frac{4}{3}k+\varepsilon)\rfloor}, 
    \end{equation}
    Finally, \Cref{eq:qtime} is improved to
\begin{equation}
    O^\ast\left(2^{\left(1-\frac{\varepsilon}{\frac{2}{3}k+\varepsilon}\right)\frac{n}{2}}\varepsilon^{-1}\right),
\end{equation}
and similarly \Cref{eq:qtimeprep} to 
\begin{equation}
    O^\ast\left(2^{\left(1-\frac{\varepsilon}{\frac{4}{3}k+\varepsilon}\right)\frac{n}{2}}\,\varepsilon^{-1}\right).
\end{equation}

\section{C. Producing low energy guiding states}

We next show that the state produced in \Cref{eqn:rho} constitutes a guiding state for the low energy space of $H$, in the following formal sense.
\begin{lemma}\label{l:guiding}
    Let $\rho$ be the $n$-qubit state in \Cref{eqn:rho}, such that $\Tr(\rho H)\leq E_0(H)+\varepsilon\mathcalE$. For any $K\geq \epsilon M$, define $\Pi_K$ as the projector onto the span of eigenvectors of $H$ with eigenvalue less than or equal to $\E_0(H)+ K$. Then, $\Tr(\rho\Pi_K)\geq 1-\frac{\epsilon M}{K}$.
\end{lemma}
\begin{proof}
    Defining $\Pi_L:= I-\Pi_{K}$, 
    \begin{align}
        \Tr(\rho H) &= \Tr((\Pi_K + \Pi_L)\rho(\Pi_K+\Pi_L)H)\\
        &=\Tr(\Pi_K\rho\Pi_K H)+ \Tr(\Pi_L\rho\Pi_L H)\\
        &\geq E_0(H)+K\Tr(\rho\Pi_L),\label{eqn:lower}
    \end{align}
    where the second statement follows since $\Pi_K$ and $\Pi_L$ partition the eigenvectors of $H$ by definition, and thus $\Pi_K H \Pi_L=0$, and the third since $\Pi_L H\Pi_L \succeq K I$. Combining \Cref{eqn:lower} with the assumption $\Tr(\rho H)\leq E_0(H)+\varepsilon\mathcalE$ yields the claim.
\end{proof}
To apply \Cref{l:guiding}, let $p$ be an arbitrary polynomial with $p(n)\geq 1$, and set $K=(1+\frac{1}{p(n)-1})\epsilon M$. By \Cref{l:guiding}, 
\begin{align}\label{eqn:overlap}
    \Tr(\rho\Pi_K)\geq 
    \frac{1}{p(n)}, 
\end{align}
i.e. $\rho$ has at least inverse polynomial overlap onto eigenvectors of $H$ of energy at most 
\begin{align}\label{eqn:loss}
    E_0(H)+K = E_0(H)+\epsilon M + \frac{\epsilon M}{p(n)-1}.
\end{align} 
Setting $p(n)\gg \epsilon M$, the term $\epsilon M/(p(n)-1)$ in \Cref{eqn:loss} asymptotically vanishes, at the cost of polynomially scaling down the overlap in \Cref{eqn:overlap}. The latter, however, is not a concern, as the goal of a guiding state is often to extract a state supported solely on the low energy space, i.e. on $\Pi_K$ in our case. To extract such a state, for example, one can apply Quantum Phase Estimation to $\rho$, whose success probability scales as $\sim 1/p(n)$. Thus, $O(p(n))$ copies of $\rho$ suffice to project successfully onto the space $\Pi_K$ with constant probability. This additional $p(n)$ does not alter our runtime bound in \Cref{eq:qtimeprep}, as it disappears into the $O^*$ notation used therein.

\section{D. Relation with Hirsch's work}

Let $f$ be an instance of $k$-SAT with $n$ variables and $m$ clauses. For an assignment $z\in\{0,1\}^n$,  write $\#(f,z)$ the number of clauses in $f$ satisfied by $z$, and write 
\[
m^\ast=\max_{z\in\{0,1\}^n}\#(f,z)\,.
\]
Hirsch's algorithm \cite{Hirsch03} finds an assignment $z$ such that
\begin{equation}\label{eq:hirsch}
\#(f,z)\ge m^\ast-\varepsilon 
m
\end{equation}
by
selecting a random assignment and trying to improve it by repeatedly flipping a value of a variable chosen randomly from an unsatisfied clause (again chosen randomly). Its complexity (see \cite[Theorem 1]{Hirsch03}) is 
\[
O^\ast\left(\left(2-\frac{2\varepsilon}{\varepsilon+k+\varepsilon k}\right)^n\right)= O^\ast\left(2^{\left(1+\log_2\left(1-\frac{\varepsilon}{\varepsilon+k+\varepsilon k}\right)\right)n}\right)\,.
\]
Since $1/\ln(2)>1.44$,
when $\varepsilon$ is small enough this bound becomes
 \[
 O^\ast\left(2^{\left(1-\frac{1.44\varepsilon}{\varepsilon+k+\varepsilon k}\right)n}\right)\,.
 \]

Our approach gives an even simpler classical algorithm for computing approximately optimal solutions, based on sampling, with similar complexity. Indeed, Theorem~\ref{thm:energyestimate} immediately gives a classical (randomized) algorithm working in time
\[
O^\ast\left(2^{\left(1-\frac{\varepsilon}{\varepsilon+k}\right) n}\right)
\]
that finds with high probability an assignment $z\in\{0,1\}^n$ satisfying \Cref{eq:hirsch}. 

\end{document}